\documentclass[aps,pra,twocolumn,showpacs,nofootinbib,floatfix]{revtex4}
\usepackage{graphicx}
\usepackage{times}
\usepackage{amssymb, amscd, amsmath,latexsym, amscd,amsfonts, amsthm, color,subfigure,hyperref}

\newtheorem{theorem}{\bf Theorem}[section]

\newtheorem{definition}[theorem]{\textsl{\bf Definition}{}}
\newtheorem{lemma}[theorem]{\bf Lemma}

\newcommand{\eq}[1]{Eq.~(\ref{#1})}

\begin{document}

\date{March 21, 2011}

\author{Peter Bowles}
\thanks{Corresponding author. \\ Electronic address: {pmxpb2@nottingham.ac.uk}}
\author{M\u{a}d\u{a}lin Gu\c{t}\u{a}}
\author{Gerardo Adesso}
\affiliation{{School of Mathematical Sciences, University of Nottingham, University Park, Nottingham NG7 2RD, United Kingdom}}

\title{Asymptotically optimal purification and dilution of mixed qubit and Gaussian states}

\begin{abstract}
Given an ensemble of mixed qubit states, it is possible to increase the purity of the constituent states using a procedure known as state purification. The reverse operation, which we refer to as dilution, produces a larger ensemble, while reducing the purity level of the systems. In this paper we find asymptotically optimal  procedures for purification and dilution of an ensemble of independently and identically distributed mixed qubit states, for some given input and output purities and an asymptotic output rate. Our solution involves using the statistical tool of local asymptotic normality, which recasts the qubit problem in terms of attenuation and amplification of a single-mode displaced Gaussian state. Therefore, to obtain the qubit solutions, we must first solve the analogous problems in the Gaussian setup. We provide full solutions to all of the above, for the (global) trace-norm figure of merit.
\end{abstract}


\maketitle

\section{Introduction}

When implementing any quantum information protocol, the states we wish to employ and manipulate are inevitably affected by decoherence effects, which diminish their purity and consequently their resource power. There exist several
well-established methods to protect against such undesirable factors:  strengthening the entanglement resource using distillation methods \cite{BBP} or employing a quantum error correction scheme \cite{S} to encode our `fragile' states into some larger, more unyielding system. The method we study in this paper is that of state {\it purification} \cite{CEM,KW}, a procedure which takes as input an ensemble of identical copies of an arbitrary (unknown) state and produces as output a smaller ensemble of identical states with higher purity. This can be seen as a special case of the more general problem of inverting the effect of a noisy channel on ensembles of states, the channel being the depolarising one in the present study.

There already exists several theoretical results for purification of $n$ i.i.d. (independently and identically distributed) mixed qubits, notably Refs.~\cite{CEM, KW}, where optimal purification algorithms for various formulations of the purification problem are provided. Purification of an ensemble of mixed qubit states has also been found to occur in the context of `superbroadcasting' \cite{DMP}, an $n \rightarrow m$ cloning procedure which can actually result in purified clones for $n\geq4$ and sufficiently mixed input states (the noise present is merely shifted from local states into correlations between output states). For $n\geq m$, superbroadcasting is actually equivalent to the optimal purification procedure of \cite{CEM}.
Experimentally, purification has been achieved in \cite{RDC}, which implemented the methodology of \cite{CEM} and demonstrated optimal purification for the case $n=2$.

Beyond the {\it entanglementology} (phenomenology of entanglement), judging the performance of a purification protocol requires a figure of merit (FoM) which measures the departure from the ideal transformation. Two types of FoM have been considered in the literature, with very different results. The {\it local} FoM is built upon the comparison of the {\it reduced} states of individual output systems with the target state.  In this case, a complete reversal of the depolarising channel may be obtained asymptotically with the size of the input ensemble, and with arbitrarily high output rate $m/n$ \cite{KW}. The {\it global} FoM compares the {\it joint} state of the output with that of a product of independent target states. This is a more demanding criterion. For example if the output systems are independent and identically prepared then the global fidelity scales as $F(n,m)= F_n^m$ where $F_n<1$ is the fidelity of an individual output state with respect to the target state. Indeed, it has been shown \cite{KW} that no protocol can achieve asymptotic purification $\big(F(n,m) \to 1 \big)$ to {\it pure} target states at a finite rate $m/n$. The global figure or merit is relevant whenever we deal with the collective state of the output rather than the individual constituents, as in the case of state transfer between atomic ensembles and light. Additionally, it can serve as a "measure of correlations" when the individual constituents of the output states are known to be exactly in the target state, as in superbroadcasting. This hypothesis will however not be pursued in this paper.

The above no-go theorem motivates us to consider the question whether the depolarising channel can be reversed with a positive asymptotic output rate, when the target states (i.e. the states prior to applying the depolarising channel) are {\it mixed}. We show that this is indeed possible, and compute the maximal purification rate for given input and target purity, and the optimal FoM for approximate purification at a fixed rate which is higher than the maximal one.

 We also consider the opposite  process of {\it dilution} in which, starting from an ensemble of $n$ identically prepared states, we produce a {\it larger} ensemble consisting of $m$ independent,  but more mixed states.
Dilution shares similarities with the process of optimal $n \rightarrow m$ quantum cloning \cite{OQC}, but while in cloning the rate $m/n$ is fixed, and one aims at generating clones as close as possible to the input states (with respect to a local or a global FoM), in a dilution procedure we set a target level of output purity and look for the optimal rate for generating such target states.
\begin{table*}[tb]
\begin{tabular}{c||c||c}
  \hline \hline
    & qubit problem & Gaussian problem \\ \hline \hline
    state model &
    $\begin{array}{c}\mbox{ensemble of $n$ i.i.d mixed qubits } \\  \rho_{\textbf{r}}^{\otimes n} \mbox{ [Eq.~(\ref{QubIN})]} \\  \mbox{$n \gg 1$:  number of copies} \\ \mbox{$\textbf{r}$ with $\|\textbf{r}\| \le 1$: Bloch vector;}  \end{array}$
    & $\begin{array}{c}\mbox{single-mode displaced Gaussian state} \\  \Phi_\alpha^s \mbox{ [Eq.~(\ref{GaussIN})]} \\ \mbox{$\alpha \in \mathbb{C}$: displacement;} \\ \mbox{$s \in (0,1)$:  purity parameter} \end{array}$ \\
  \hline
  input & $\rho_{\textbf{r}_0+\textbf{u}/\sqrt{n}}^{\otimes n}$& $\Phi_\alpha^{s_1}$ \\
  \hline
  target & $\rho_{\lambda \textbf{r}_0+k\textbf{u}/\sqrt{m}}^{\otimes m}$ & $\Phi_{k \alpha}^{s_2}$ \\
  \hline
  procedure & $\begin{array}{c}\mbox{purification} \\ \lambda>1, \ m<n\end{array}$ & $\begin{array}{c}\mbox{attenuation} \\ s_2 < s_1, \ k<1\end{array}$ \\
  \hline
  procedure & $\begin{array}{c}\mbox{dilution} \\ \lambda < 1, \ m>n \end{array}$ & $\begin{array}{c}\mbox{amplification} \\ s_2 > s_1, \ k>1\end{array}$ \\
    \hline
   \hline
\end{tabular}
\caption{Summary of the notation adopted in the present paper. For the qubit problem, we aim at optimising the output-vs-input rate $m/n$ (maximising it for purification, and minimising it for dilution) at given input Bloch vector $\textbf{r}_0+ \textbf{u}/\sqrt{n}$ and scale factor $\lambda$. For the corresponding Gaussian problem, we aim at finding the maximal value of the displacement ratio $k$, such that attenuation or amplification can be realised perfectly, for given target temperature parameters $s_1$ and $s_2$ and unknown displacement $\alpha$. The framework of local asymptotic normality provides a rigorous link between the two problems, as for $n \gg 1$ the local Bloch vector $\textbf{u}$ is mapped into the displacement $\alpha$ of a single-mode coherent thermal state.}\label{table}
\end{table*}

 In deriving the asymptotic results, the key mathematical tool is that of local asymptotic normality (LAN), a fundamental `classical' statistics technique \cite{LC} which was recently extended to the context of quantum statistical models \cite{GK,GJK,GK2,GJ}.  In the quantum case, LAN dictates that the collective state of $n$ i.i.d. quantum systems, can be approximated by a joint Gaussian state of a classical and a quantum continuous variable (CV) systems. This has been used to derive asymptotically optimal state estimation strategies for mixed states of arbitrary finite dimension \cite{GK2}, and also in finding quantum teleportation benchmarks \cite{GBA} and optimal quantum learning procedures \cite{GKot} for multiple qubit states. The general strategy is to recast statistical problems
involving $n$ i.i.d. quantum systems into the simpler setting of Gaussian states. The optimal solution for the corresponding Gaussian problems can then be used to construct asymptotically optimal procedures for the original one. In section \ref{sec3} we sketch how this could be physically implemented, and more details can be found in \cite{GJK}.

Following this methodology, we transform the qubit purification and dilation problems into those of optimal {\it attenuation} and {\it amplification} for  a one-mode CV system in a Gaussian state, together with a classical real-valued Gaussian variable, both with known variance but unknown means. In attenuation we reduce the variance of a displaced Gaussian state, at the price of simultaneously reducing its amplitude, while in amplification we increase the amplitude, as well as the variance. For both problems we use a FoM based on maximum trace-norm distance, and show that the optimal attenuation channel is obtained by applying a beamsplitter, while the optimal amplification is implemented by a non-degenerate parametric amplifier. A similar scheme for the  attenuation of Gaussian CV states has been proposed and experimentally implemented in \cite{AFF}. Parametric amplification has been investigated in \cite{HM,CC,CDG}, and  demonstrated experimentally in \cite{OUP}. In particular, the same amplifier is optimal for a FoM based on the minimum amount of added noise \cite{HM,CC}. However, whilst these transformations are well known candidates for our protocols, to the best of our knowledge a proof of their optimality with respect to the FoM chosen in this paper  had not been obtained in the literature. Our proof relies on a covariant channels optimisation technique developed in \cite{GM,GBA}.  We find that for given input and output purity parameters, there exists a range of values for the ratio $k$ between output and input displacement, such that attenuation or amplification can be realised perfectly, and we compute the maximal (optimal) value $k_0$, as a function the two purities.  In the parameter range where the procedures cannot be accomplished perfectly, we give the exact expression for the optimal FoM.

A schematic summary of the problems addressed in this paper is provided in Table \ref{table}. The paper is organised as follows. In Section~ \ref{sec2} we formulate and solve the two quantum Gaussian problems, and the corresponding classical one. In Section~\ref{sec3} we use this result in conjunction with LAN to find asymptotically optimal purification and amplification channels for states of $n$  i.i.d. mixed qubits.
We draw our concluding remarks in Section~\ref{sec4}.
The proofs are collected in Appendix A.

\section{Optimal attenuation and amplification of Gaussian states}
\label{sec2}
\subsection{Classical Case}
\label{sec2a}
Before we move onto the quantum case, it is instructive and relevant to consider the corresponding problems for classical random variables. In the classical scenario, the analogue of `attenuation' (`amplification') is a procedure which reduces (increases) the mean and variance of a given random variable. The analogue to our quantum problem would then be to find a transformation $K$ which maps a real-valued normally distributed random variable $X\sim N(u,V_1)$ of arbitrary mean $u$ and fixed variance $V_1$,  into a variable $Y\sim N(ku, V_2)$ such that the risk
\begin{equation}
R_{\rm max}(K; V_{1}, V_{2}, k) =2 \sup_u\|K\big(N(u, V_1)\big)-N(ku,V_2)\|_{\text{tv}}
\label{cfom}
\end{equation}
is minimised. Here $k$ represents a fixed constant, where $0<k<1$ means attenuation and $k>1$ means amplification of the Gaussian variable $X$, and we choose the interesting case where $V_1>V_2$ in the case of attenuation, and  $V_1<V_2$ for amplification. The notation $\|\mathbb{P}-\mathbb{Q}\|_{\text{tv}} = \sup\{|\mathbb{P}(A)-\mathbb{Q}(A)|:A\in\mathcal{F}\}$, for the $\sigma$-algebra $\mathcal{F}$, represents the total variation distance between the probability distributions
$\mathbb{P}$ and $\mathbb{Q}$ which reduces to one-half of the $L_{1}$-distance between their probability densities in the case of mutually absolutely continuous distributions \cite{Torgersen}.

The solutions of both classical and quantum versions of this problem rely on the notion of `covariance'. Consider the transformation
\begin{equation}
X\mapsto K(X)=kX+Z
\label{ccov}
\end{equation}			
where $X$ and $Z$ are independent random variables, $Z$ having a fixed variance and vanishing mean. Such a (classical) channel is covariant,
in the sense that
\begin{equation}
K(X+C) = K(X) +kC
\end{equation}
for any constant $C$. Such transformations can be shown to not only minimise (\ref{cfom}), but also to render it independent of expectation so that the FoM becomes
\begin{equation*}
{R}_{\rm max}(K; V_{1},V_{2},k)=2 \|K\big(N(0, V_1)\big)-N(0,V_2)\|_{\text{tv}}.
\end{equation*}
It is easy to see that if
\begin{equation}
k\leq  k^{(c)}_{0} (V_{1}, V_{2}):= \sqrt{\frac{V_2}{V_1}}
\label{cr}
\end{equation}
then the target distribution can be achieved exactly, with the appropriate amount of Gaussian noise in the variable $Z$. As we shall see in the next section, there exists an analogous range (\ref{koeqn}) for the quantum Gaussian transformation.

As for the case $k> k_{0}^{(c)} (V_{1}, V_{2})$, it can be shown \cite{Torgersen} that the optimal choice for $Z$ in \eqref{ccov} is $Z=0$, as one would expect, so that the optimal figure of merit is
\begin{eqnarray}
&&R_{\rm minmax}(V_{1},V_{2},k):= \inf_{K}R_{\rm max}(K;V_{1},V_{2},k) \nonumber\\
&&=\int{\bigg|\frac{1}{\sqrt{2\pi k^2 V_1}}e^{-\frac{x^2}{2k^2V_1}}-\frac{1}{\sqrt{2\pi V_2}}e^{-\frac{x^2}{2V_2}}\bigg|dx}.
\label{clrisk}
\end{eqnarray}

Henceforth, we will denote by $K^*$ the optimal transformation for the two cases discussed above.
\begin{figure}[t]
\includegraphics[width=8.5cm]{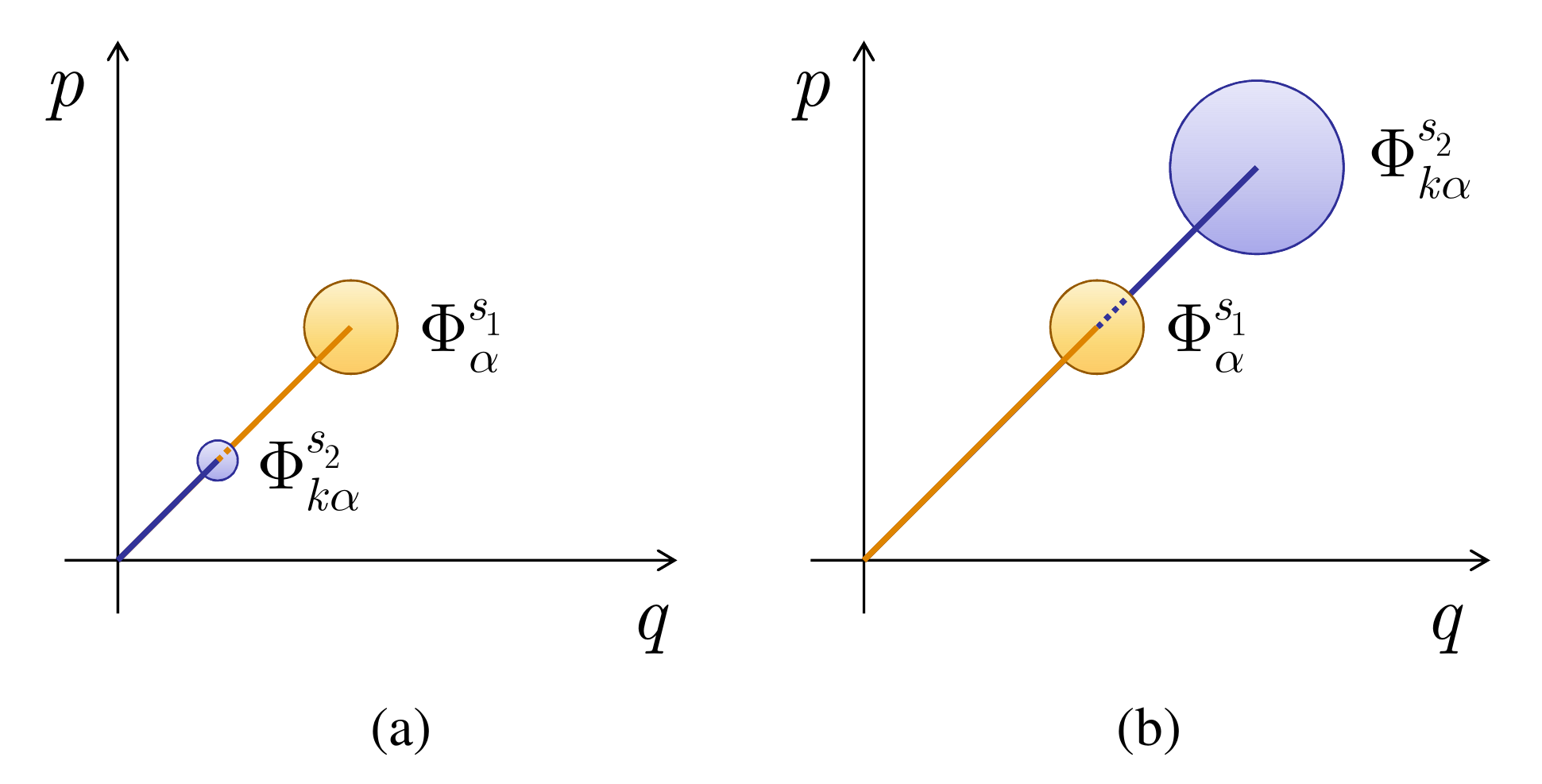}
\caption{(Color online) Schematic phase-space diagram for (a) attenuation and (b) amplification of a displaced Gaussian state $\Phi^{s_1}_\alpha$. }
\label{figps}
\end{figure}

\subsection{Quantum Case}
\label{sec2b}

In this Section we consider the following: given a Gaussian state $\Phi_\alpha$ of a one-mode CV quantum system, with known covariance and unknown displacement $\alpha$, we would like to optimally attenuate (amplify) it, that is transform it into a state with smaller (greater) covariance and displacement $k\alpha$, with the largest possible proportionality constant $k$. Let
$$
W_\alpha := \exp(\alpha a^\dagger-\bar{\alpha}a)
$$
denote the Weyl operators where $\alpha\in\mathbb{C}$ and $a$, $a^\dagger$ the creation and annihilation operators satisfying
$[a,a^\dagger]=\mathbf{1}$ and
$$
a|n\rangle=\sqrt{n}|n-1\rangle, \qquad n\geq 0,
$$ where $\{|n\rangle\}_{n\geq0}$ is the Fock basis of the Hilbert space $\mathcal{H}$. For $0<s<1$ we denote by $\Phi^s$ the centred, phase invariant Gaussian state
\begin{equation}\label{eq.thermal}
\Phi^s = (1-s)\sum^\infty_{n=0} s^n|n\rangle\langle n|,
\end{equation}
and by displacing it we obtain the family of Gaussian states
\begin{equation}
\Phi^s_\alpha := W_\alpha \Phi^s W^\dagger_\alpha.
\label{GaussIN}
\end{equation}
Given two different mixing parameters $s_1>s_2$ $(s_1<s_2)$ and a positive parameter $k<1$ $(k>1)$ we would like to find the optimal attenuation (amplification) channel which maps the state $\Phi^{s_1}_\alpha$ close to the state $\Phi^{s_2}_{\alpha}$ for an arbitrary  displacement $\alpha$ (see Fig.~\ref{figps}). For any channel $P:\mathcal{T}_1(\mathcal{H}) \rightarrow \mathcal{T}_1(\mathcal{H})$ we define the FoM called the maximum risk
\begin{equation}
R_{\rm max}(P; s_1,s_2,k) = \sup_{\alpha\in\mathbb{C}}\|P(\Phi^{s_1}_\alpha)-\Phi^{s_2}_{k\alpha}\|_1
\end{equation}
and the minimax risk
\begin{equation}
R_{\rm minmax}(s_1,s_2,k) = \inf_P R_{\rm max}(P; s_1,s_2,k).
\label{mmx}
\end{equation}
A channel is called `minimax' if its maximum risk is equal to the minimax risk. We will show that (up to a trivial adjustment for a certain range of $k$'s) the optimal solutions to the attenuation and amplification problems are, respectively, the beamsplitter and parametric amplifier.

\begin{figure*}[tbh]
\subfigure[\label{figpur2d}]
{\includegraphics[width=6.2cm]{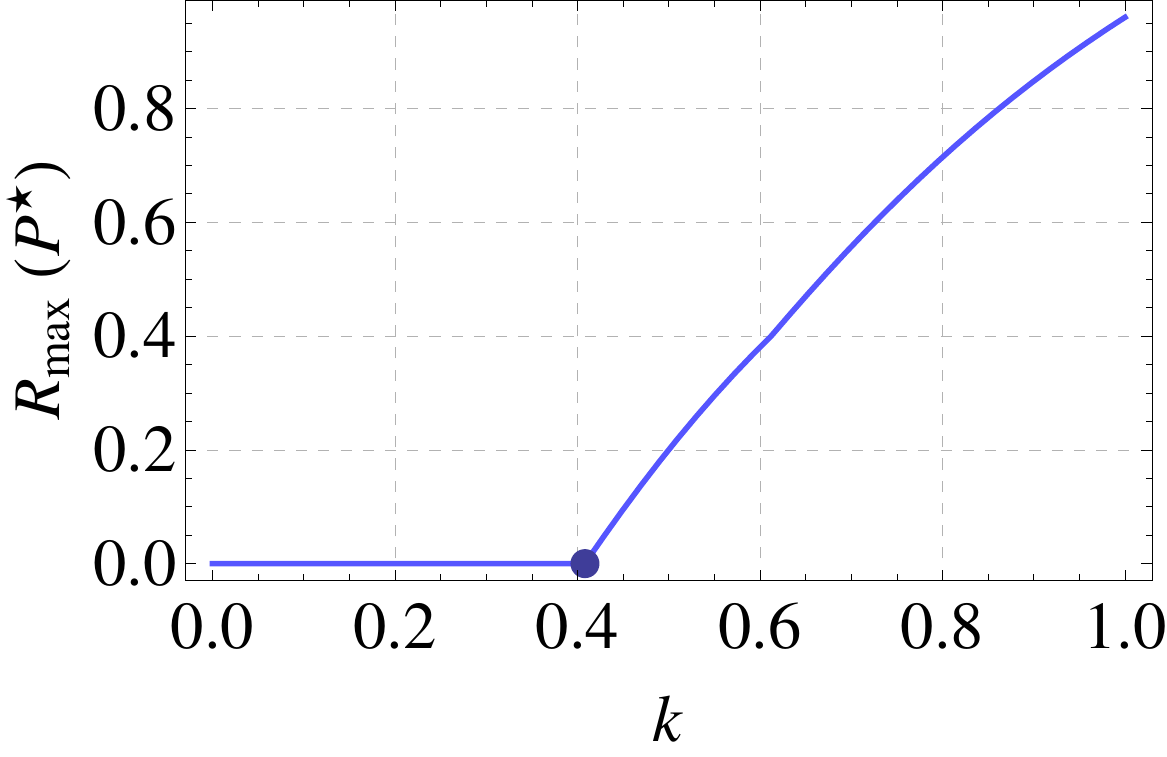}} \hspace{15mm}
\subfigure[\label{figamp2d}]
{\includegraphics[width=6.2cm]{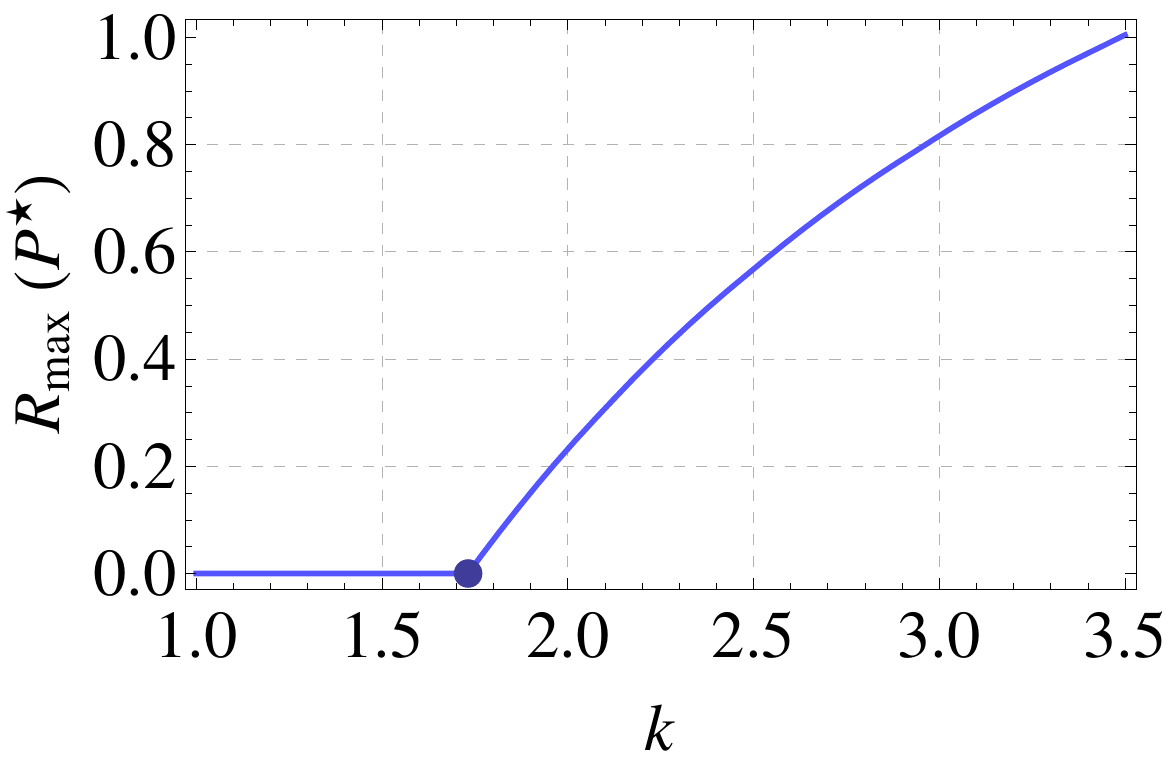}} \\
\subfigure[\label{figpur3d}]
{\includegraphics[width=8.0cm]{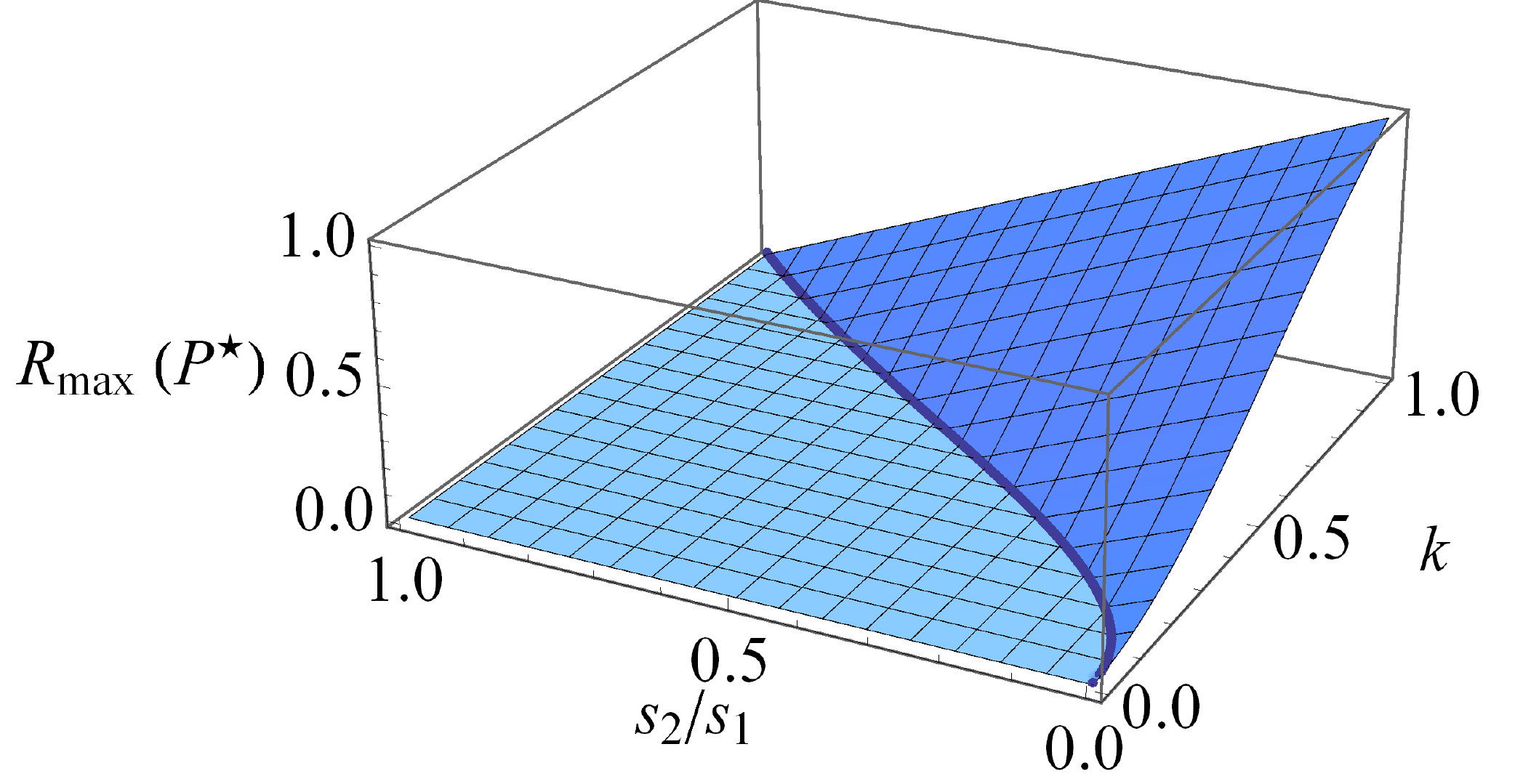}} \hspace{3mm}
\subfigure[\label{figamp3d}]
{\includegraphics[width=8.0cm]{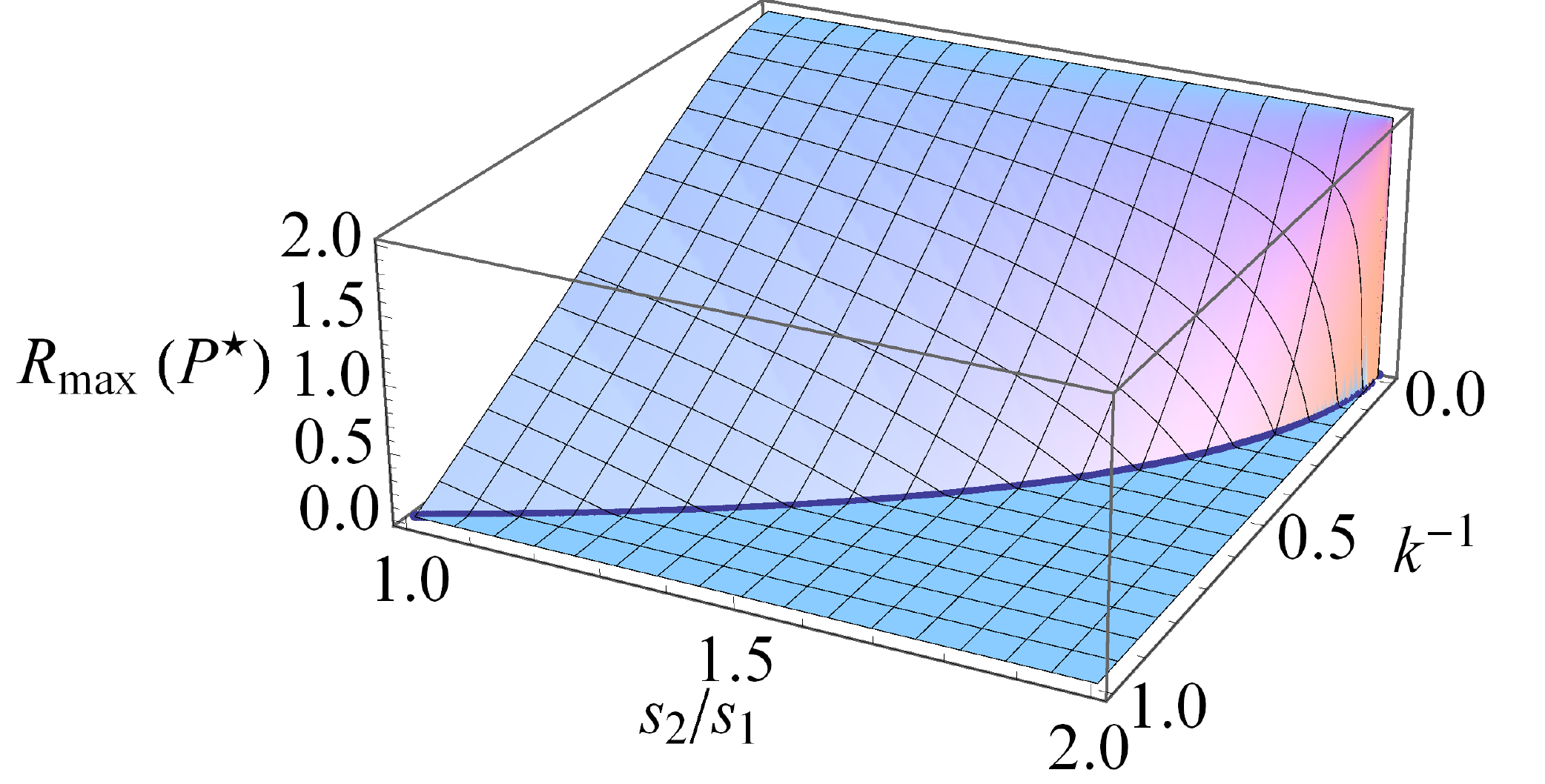}}
\caption{(Color online)
FoM's for optimal attenuation (left) and optimal amplification (right) of displaced Gaussian states.
Top row: (a) Plot of the minimax risk ${R}_{\max}$ [Eq.~(\ref{Rmaxpur})] versus  $k$ for the optimal attenuation procedure $P^{\star}$, where $s_1=0.8$ and $s_2=0.4$; the optimal value $k_0$ is highlighted with a big dot on the graph. (b) Plot of the minimax risk ${R}_{\max}$ [Eq.~(\ref{Rmaxpur})] versus  $k$ for the optimal amplification procedure $P^{\star}$, where $s_1=0.4$ and $s_2=0.8$; the optimal value $k_0$ is highlighted with a big dot on the graph.
Bottom row: 3D Plots of the minimax risk ${R}_{\max}$ versus the parameter $k$ and the output/input temperature ratio $s_2/s_1$ with $s_1=0.5$, for (c) the optimal attenuation procedure, and (d) the optimal amplification procedure; in both plots, the thick curve depicts $k_0$ as a function of $s_2/s_1$.}
  \label{figpuramp}
\end{figure*}

We start by defining a specific channel denoted in both cases $P^{\star}$, then show that it is optimal and compute the minimax risk. For $s_1>s_2$ and $k<1$, the attenuation channel is implemented by the action of a beamsplitter with reflectivity $k$ acting on an input mode $a$ prepared in a state $\Phi^{s_1}_\alpha$, and a second ancillary mode $b$ in the vacuum state $\omega=|0\rangle\langle 0|$. The output mode $c$ of the channel is \begin{equation}
c = k^2a + \sqrt{1-k^2}b.
\label{bs}
\end{equation}
For $s_2>s_1$ and $k>1$, the channel is a parametric amplifier, whose action is represented by the following transformation on the input mode $a$ and  an ancillary mode $b$ prepared in the vacuum state:
\begin{equation}
c=ka + \sqrt{k^2-1}b^\dagger .
\label{am}
\end{equation}

We note that for each pair $(s_1,s_2)$ there exists a range of parameters $k$ for which $R_{\rm minmax}(s_1,s_2,k)=0$, i.e., the procedures can be accomplished perfectly. Indeed it can be easily verified that, for $k$ given by
 \begin{equation}
k^{\rm att}_0(s_1, s_2) = \sqrt{\frac{s_2(1-s_1)}{s_1(1-s_2)}},
\quad k^{\rm amp}_0(s_1, s_2)= \sqrt{\frac{1-s_1}{1-s_2}},
\label{koeqn}
\end{equation}
the channels \eqref{bs} and respectively \eqref{am} produce exactly the target state
$\Phi^{s_{2}}_{k\alpha}$. Moreover, if $k<k_{0}$ then the output of $P^{\star}$ is the state $\Phi^{s}_{k\alpha}$ with $s<s_{2}$, and the target can be still perfectly achieved by adding an appropriate amount of Gaussian noise. For later use, when $k<k_{0}$ we will denote by the same symbol
$P^\star$ this modified channel. From now on we consider the less trivial situation $k\geq k_{0}$, corresponding to the regime where perfect amplification or attenuation are impossible. We then state the following theorem and lemma, whose proofs are given in Appendix A:
\begin{theorem}\label{theorem21}
If $k<k_0$ then the minimax risk for attenuation (amplification) is zero.
If $k\geq k_0$, the minimax procedure is $P^\star$, i.e. the beamsplitter (\ref{bs}) in the case of attenuation, or the parametric amplifier (\ref{am}) in the case of amplification:
\begin{equation}
R_{\rm max}(P^\star;s_{1},s_{2},k)= R_{\rm minmax}(s_{1},s_{2},k).
\label{minmax}
\end{equation}
\end{theorem}
\begin{lemma}\label{lemma22}
If $k<k_0$, then the minimax risk for attenuation (amplification) is given by
\begin{equation}
R_{\rm minmax}(s_{1},s_{2},k)=
2(\tilde{s}^{m_{0}+1} - s_{2}^{m_{0}+1}) ,
\label{Rmaxpur}
\end{equation}
where $m_0$ is the integer part of
$$
 \ln [ (1-\tilde{s})/(1-s_{2})] /\ln(s_{2}/\tilde{s}),
$$
and $\tilde{s}$ takes the values
\begin{equation}
\tilde{s}_{\rm att} = \frac{s_1 k^2}{1-s_1 + s_1 k^2}, ~{\rm and}~
\tilde{s}_{\rm amp}= 1-\frac{1-s_1}{k^2}.
\label{tildes}
\end{equation}
in the case of attenuation and respectively amplification.

\end{lemma}
The risk for both processes is plotted in Fig.~\ref{figpuramp} [(a)-(d)]. In Figure \ref{figk0} we plot $k_0$ for both processes as a function of the input and
output purity  parameters $s_1$ and $s_2$.

\begin{figure}[tb]
\subfigure[\label{figpurk0}]
{\includegraphics[width=4.1cm]{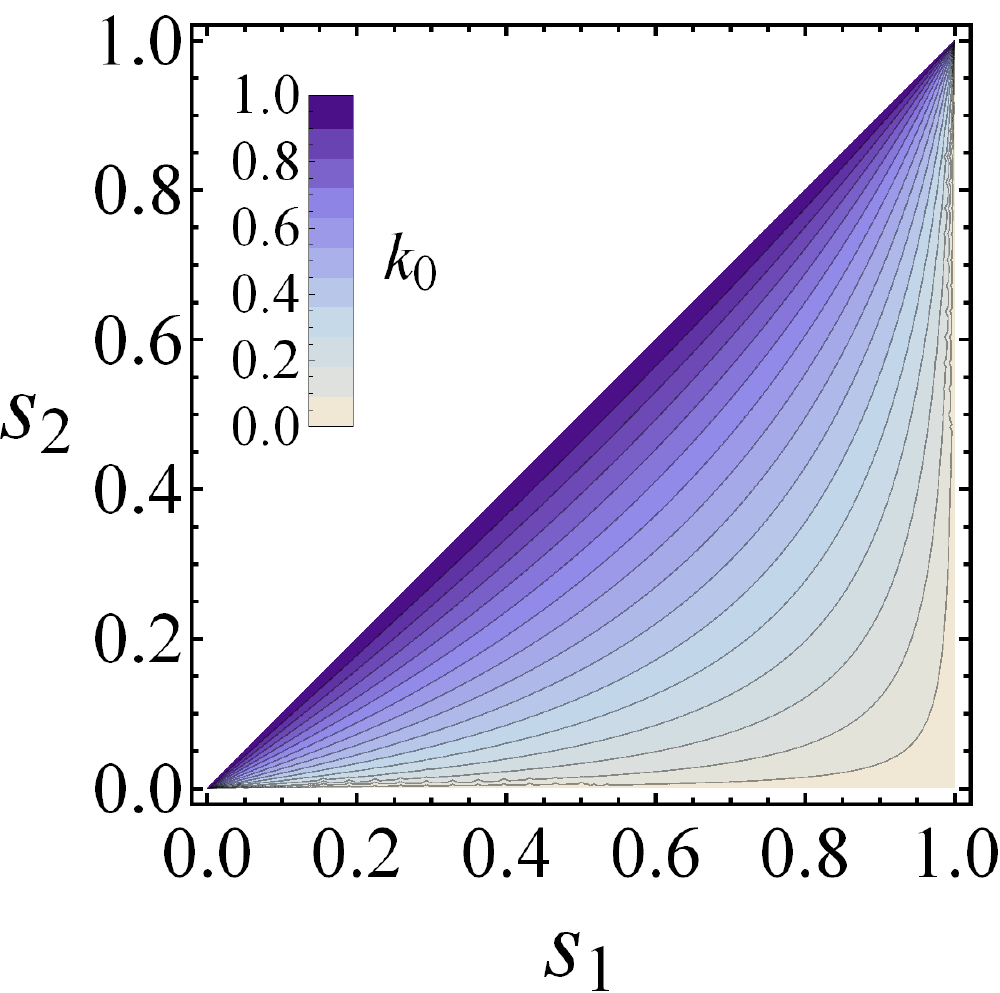}} \hspace{2mm}
\subfigure[\label{figampk0}]
{\includegraphics[width=4.1cm]{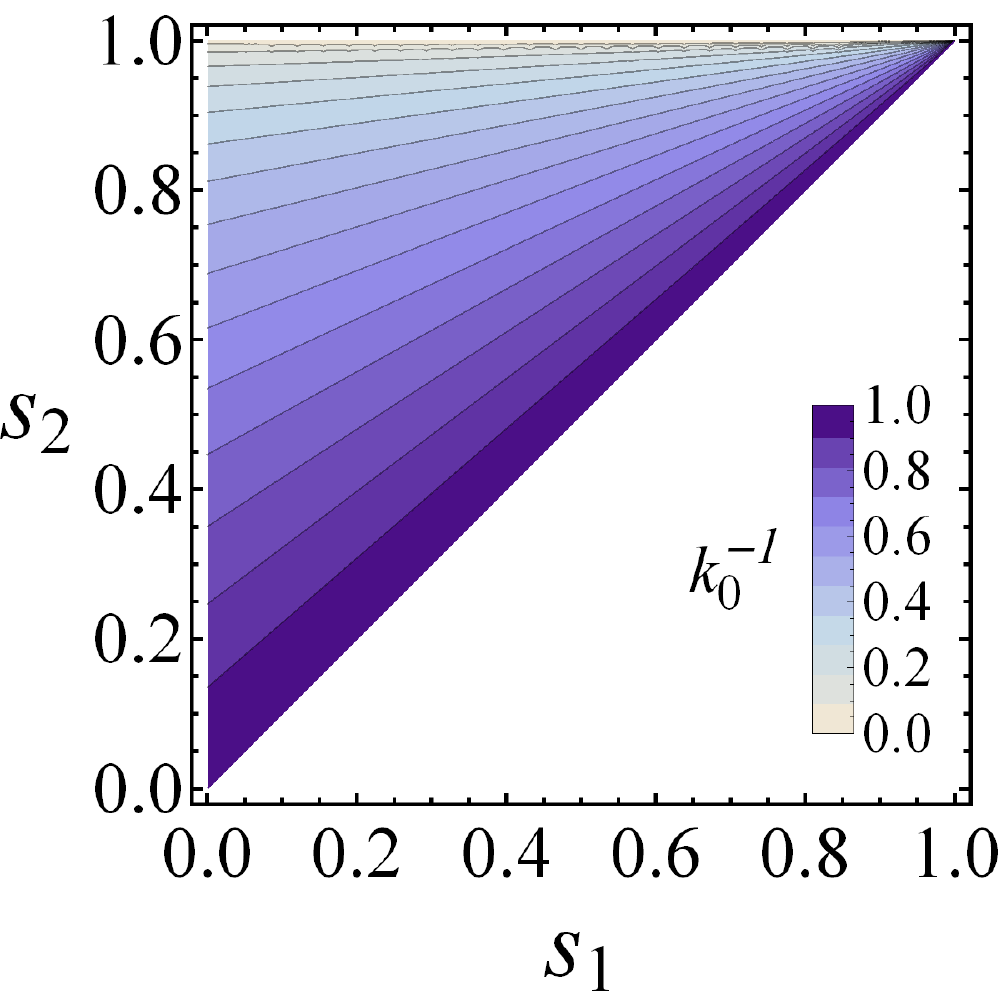}}
\caption{(Color online) Contour plots of (a) $k_0$ versus the purity parameters $s_1$ and $s_2$ for the optimal attenuation procedure, and (b) ${k_0}^{-1}$ versus the purity parameters $s_1$ and $s_2$ for the optimal amplification procedure.}
\label{figk0}
\end{figure}

\section{Asymptotically optimal purification and dilution for ensembles of qubits}
\label{sec3}
We turn now to the problem of finding optimal purification and dilution schemes for ensembles of identical qubits. We denote by $\rho_{\textbf{r}}$ the qubit state with Bloch vector $\textbf{r}=(r_{x},r_{y},r_{z})$
\begin{equation}\label{QubIN}
\rho_{\textbf{r}}=\frac{1}{2}(\mathbf{1}+\textbf{r}\boldsymbol{\sigma})\,,
\end{equation} where
$\textbf{r}\boldsymbol{\sigma}=r_x\sigma_{x} +r_y\sigma_{y} +r_z\sigma_{z}$ and $\sigma_{i}$ are the Pauli matrices. We are given $n$ identical qubits prepared in the state $\rho_{\textbf{r}}$ and we would like to produce $m$ identical qubits in the state $\rho_{\lambda \textbf{r}}$ with $m$ as large as possible, for a fixed positive parameter  $\lambda$.
When $\lambda >1$, the aim is to ``purify'' the state, and when $0<\lambda <1$ we want to ``dilute'' the state with the benefit of obtaining more copies.
Clearly, for purification the output state is physical only if $\lambda$ satisfies
$\lambda\| \textbf{r}\|\leq 1$. This can be achieved by letting $\lambda$ depend on $\textbf{r}$, or by restricting to those input states which satisfy the property. To illustrate the latter, suppose we would like to reverse the action of  the
depolarising channel
$$
C :\rho_{\textbf{r}}\mapsto \rho_{\textbf{r}/\lambda},
$$
then the input states of the purification channel automatically satisfy the requirement. As to the former, our asymptotic analysis will produce a {\it local} FoM which only depends on the value of $\lambda$ at a particular state, so for simplicity we will assume it to be constant.

A purification (dilution) procedure is a quantum channel
$$
Q_{n}:M( \mathbb{C}^{2^{n}})\to M( \mathbb{C}^{2^{m}})
$$
mapping $n$-qubit states to $m$-qubit states, and its performance is measured by the  FoM (risk)
\begin{equation}
R(Q_n;\textbf{r},\lambda):=\|Q_{n}(\rho_{\textbf{r}}^{\otimes n})- \rho^{\otimes m}_{\lambda\textbf{r}}\|_1.
\end{equation}
Note that this is a global rather than a local risk, in the sense that it measures
the distance between the output and the {\it joint} product state, instead of comparing their restrictions to each single system.  Note also that the risk at a fixed point $\textbf{r}$ can always be made equal to zero by simply preparing the target state for that point. To take into account the {\it overall} performance of a procedure, one can either integrate the risk with respect to a prior distribution over states (Bayesian statistics) or take the maximum over all states  (frequentist statistics). We adopt the latter viewpoint, and in addition we will consider a more refined version of the maximum risk called
{\it local maximum risk} around $\textbf{r}_0$
\begin{equation}
R_{\rm max}(Q_n;\textbf{r}_0,\lambda):=\sup_{\|\textbf{r}-\textbf{r}_0\|\leq n^{-\frac{1}{2}+\epsilon}} R(Q_n;\textbf{r},\lambda).
\end{equation}
In asymptotic statistics the local maximum risk is more informative that the `global'  one since it captures the behaviour of the procedure around any point in the parameter space, rather than that of the worst case. The radius of the ball over which we maximise is  slightly larger than the precision of $n^{-1/2}$ with which we can estimate the state parameters, so that the definition of the local risk does not amount to assuming any prior information about the parameter.
Indeed one can use a small sample $n^{1-\epsilon}\ll n$ of the input systems to obtain a rough estimate of the Bloch vector $\textbf{r}$ such that the obtained estimator
$\textbf{r}_{0}$ will be in a ball of size $n^{-1/2+\epsilon}$ around $\textbf{r}$, with probability converging to one as $n\to \infty$. With this additional information, one can then apply the purification (dilution) channel to the remaining systems, with no loss in the asymptotic optimal risk (see below). The local maximum risk is a standard FoM in asymptotic statistics  and it is has been used in quantum statistics in \cite{GK,GJK,GKot} to which we refer for more details, and for its relation to Bayesian methods.

Up to this point the number of input and output systems $n$ and $m$ were fixed, with $n$ considered to be large. However, for a non-trivial asymptotic analysis, $m$ should be an increasing function of $n$, more precisely we  consider the optimal purification (dilution) procedure for a fixed rate
$$
\Lambda= \lim_{n \rightarrow \infty} \frac{m(n)}{n} >0.
$$
Indeed from our fixed rate analysis it can easily be deduced that in the case of
a sub-linear dependence $m(n)=o(n)$, one can produce $m$ output copies of arbitrary purity with vanishing local maximum risk. On the other hand, by similar reasonings, one may expect that if $m(n)/n$ is unbounded, then the best strategy should be to estimate the state and reprepare $m$ independent copies of the estimator (`measure and prepare' strategy \cite{Bae}). We leave this statement as a conjecture, and from now on we will assume that the rate $\Lambda$ is given and fixed. For any sequence $\mathbf{Q}:= \{Q_{n}\}$ of procedures we define the {\it asymptotic local maximum risk} at $\textbf{r}_0$ by
\begin{equation}
R(\mathbf{Q};\textbf{r}_0,\lambda,\Lambda):=
\limsup_{n\rightarrow\infty}R_{\rm max}(Q_n;\textbf{r}_0,\lambda),
\end{equation}
and we would like to find an optimal (minimax) strategy whose asymptotic risk is equal to the {\it minimax risk}
\begin{equation}
R_{\rm minmax}(\textbf{r}_0,\lambda,\Lambda) :=
\limsup_{n\rightarrow\infty}\inf_{Q_n}R_{\rm max}(Q_n;\textbf{r}_0,\lambda).
\label{qmmx}
\end{equation}
In other words, we will answer the following question: for given purification (dilution) constant scale factor $\lambda$ and input-output rate $\Lambda$, what is the minimax risk $R_{\rm minmax}(\textbf{r}_0,\lambda,\Lambda)$ and which is the procedure that achieves it? In particular, we will find that the minimax risk is zero for a range of parameters $(\textbf{r}_0,\lambda,\Lambda)$, and we will identify the maximum value $\Lambda_0^{\rm pur}$ ($\Lambda_0^{\rm dil}$) for which the purification (dilution) can be performed with asymptotically vanishing risk. These rates are the qubit analogues of the constants $k_0$ defined in \eqref{koeqn}.

The main technical tool is the theory of {\it local asymptotic normality} (LAN) developed in \cite{GK,GJK,GJ,GK2} as an extension of a key concept from
(classical) asymptotic statistics \cite{LC,vanderVaart}. LAN means that the
joint quantum state of identically prepared (finite-dimensional) systems can be approximated in a strong sense by a quantum-classical Gaussian state of fixed variance, whose mean encodes the information about the parameters of the original state. In this way, a number of asymptotic problems can be reformulated in terms of Gaussian states, for which the explicit solution can be found, e.g. state estimation \cite{GK3}, teleportation benchmarks \cite{GBA}, quantum learning \cite{GKot}, system identification \cite{G}.  For the purposes of this paper we give a brief description of LAN for mixed qubit states. Let
$$
\rho_{\textbf{r}_0+ \textbf{u}/ \sqrt{n}}=
\frac{1}{2}\big(\textbf{1}+(\textbf{r}_0+\textbf{u}/\sqrt{n})\boldsymbol{\sigma}\big)
$$
denote a qubit state in a the neighbourhood of a fixed and known state
$\rho_{\textbf{r}_0}$, which is uniquely characterised by an  {\it unknown} local parameter $\textbf{u}$. The family of $n$-qubit states
\begin{equation}\label{eq.local.qubit.model.}
\mathcal{P}_{n}:= \left\{ \rho_{\textbf{u}}^{n}:= \rho_{\textbf{r}_0+ \textbf{u}/ \sqrt{n}}^{\otimes n} :
\| \textbf{u}\| \leq n^{\epsilon}\right\}
\end{equation}
will be called the local statistical model at $\textbf{r}_0$.
Additionally, we define a classical-quantum Gaussian model
\begin{equation}\label{eq.gaussian.model}
\mathcal{N}:=\left\{ N_{\textbf{u}}\otimes \Phi_{\textbf{u}} :\textbf{u}\in \mathbb{R}^{3}\right\}
\end{equation}
where $N_{\textbf{u}} := N(u_{z}, 1-\|\textbf{r}_{0}\|^{2})$ is a normal (Gaussian) distribution on $\mathbb{R}$
with mean $u_{z}$ and variance $1-\|\textbf{r}_{0}\|^{2}$, and
$$
\Phi_{\textbf{u}} =
W_{\alpha} \Phi^{s} W_{\alpha}^{\dagger}, \qquad
\alpha= \frac{u_{x}+iu_{y}}{2\|\textbf{r}_{0}\|},\,\,
s= \frac{1-\|\textbf{r}_{0}\|}{1+\|\textbf{r}_{0}\|},
$$
is a displaced thermal Gaussian state of a one-mode CV system (cf. Section~\ref{sec2b}) with known covariance matrix characterised by the purity parameter $s$ (with zero squeezing) and unknown means proportional to $(u_{x},u_{y})$. Now, the mathematical statement of LAN \cite{GK} is that there exist two sequences of channels
$\mathbf{T}= \{T_{n}\}$ and $\mathbf{S}=\{S_{n}\}$ with
\begin{eqnarray*}
T_{n}&:& M(\mathbb{C}^{2^{n}})\to L^{1}(\mathbb{R})\otimes \mathcal{T}_{1}\\
S_{n}&:& L^{1}(\mathbb{R})\otimes \mathcal{T}_{1}\to M(\mathbb{C}^{2^{n}})
\end{eqnarray*}
 such that
\begin{eqnarray*}
&&
\lim_{n\to \infty }\sup_{\|{\bf u} \|\leq n^\epsilon }\| T_{n} (\rho_{\textbf{u}}^{n})- N_{\textbf{u}}\otimes \Phi_{\textbf{u}}\|_{1} =0,\\
&&
\lim_{n\to \infty} \sup_{\|{\bf u} \|\leq n^\epsilon }\| \rho_{\textbf{u}}^{n}- S_{n}(N_{\textbf{u}}\otimes \Phi_{\textbf{u}} ) \|_{1} =0.
\end{eqnarray*}
In the above formulas, $\mathcal{T}_{1}$ is associated to  the trace-class operators of the CV system, and the norm-one $\| \cdot \|_1$ denotes respectively the trace-norm for the quantum part and the $L_{1}$-norm for the classical part. The physical implementation for the channels $T_n$ and $S_n$, detailed in \cite{GKJ}, is realised via a spontaneous emission coupling of the $n$ qubits to a Bosonic field, and subsequently letting the qubits 'leak' into this environment. Since there is no correlation between atoms and field, the statistical model decouples into a Gaussian state $\Phi_{\textbf{u}}$ associated to the field, and a classical statistical mixture of atoms, distributed according to $N_{\textbf{u}}$. The corresponding operations of attenuation and amplification may then be carried out in the field in the optimal way.

In our case we need to consider mixed qubit states, which means that the collective state has non-zero components in all irreducible representations of SU(2) (all values of the total spin). In fact the traces of the different blocks of given total spin form a probability distribution which (after centring and scaling) converges to the classical Gaussian component of the limit model in LAN. A typical block state of definite total spin can be mapped isometrically into the Fock space of a one mode CV system, and converges to the quantum Gaussian component of the limit model. This transfer can be implemented in principle by a creation-annihilation coupling with a Bosonic field in which the state `leaks'' after a short time. The classical part (total spin) can be measured  by coupling subsequently with another Bosonic field, and performing an indirect measurement of $L_{z}$ in the  field. Since after the first step, all blocks are in the $|j,j\rangle$ state, a measurement of $L_{z}$ is implicitly a measurement of the total spin.

The above convergence can be interpreted as follows: the quantum statistical models $\mathcal{P}_{n}$ can be mapped into the Gaussian model
$\mathcal{N}$ and vice-versa, by means of physical operations (quantum channels) with vanishing norm-one error. From the statistical point of view, in many situations this  convergence is strong enough to allow us to map a statistical problem concerning the model $\mathcal{P}_{n}$ to a similar  one concerning the simpler model $\mathcal{N}$.

In the case of purification or dilution of qubits, the mapping into a Gaussian problem is illustrated in the diagram below. We first give a detailed
description of the steps involved, and then prove that our procedure is optimal
(asymptotically minimax).

\begin{equation}\begin{CD}
  \rho^{\otimes n}_{\textbf{r}_0+\frac{\textbf{u}}{\sqrt{n}}} @>Q^{\star}_n>> \rho_{\lambda\textbf{r}_0+\frac{\textbf{u}^{\prime}}{\sqrt{m}}}^{\otimes m} \\
   @V{T_n}VV @AA{S_m}A      \\
    N_{\textbf{u}}\otimes  \Phi^{s_1}_{\textbf{u}} @>K^{\star}\otimes P^{\star}>>
    N_{\textbf{u}'}\otimes\Phi^{s_2}_{\textbf{u}'} \\
   \end{CD}\label{cd}\end{equation}

\begin{figure}[t]
\includegraphics[width=8.5cm]{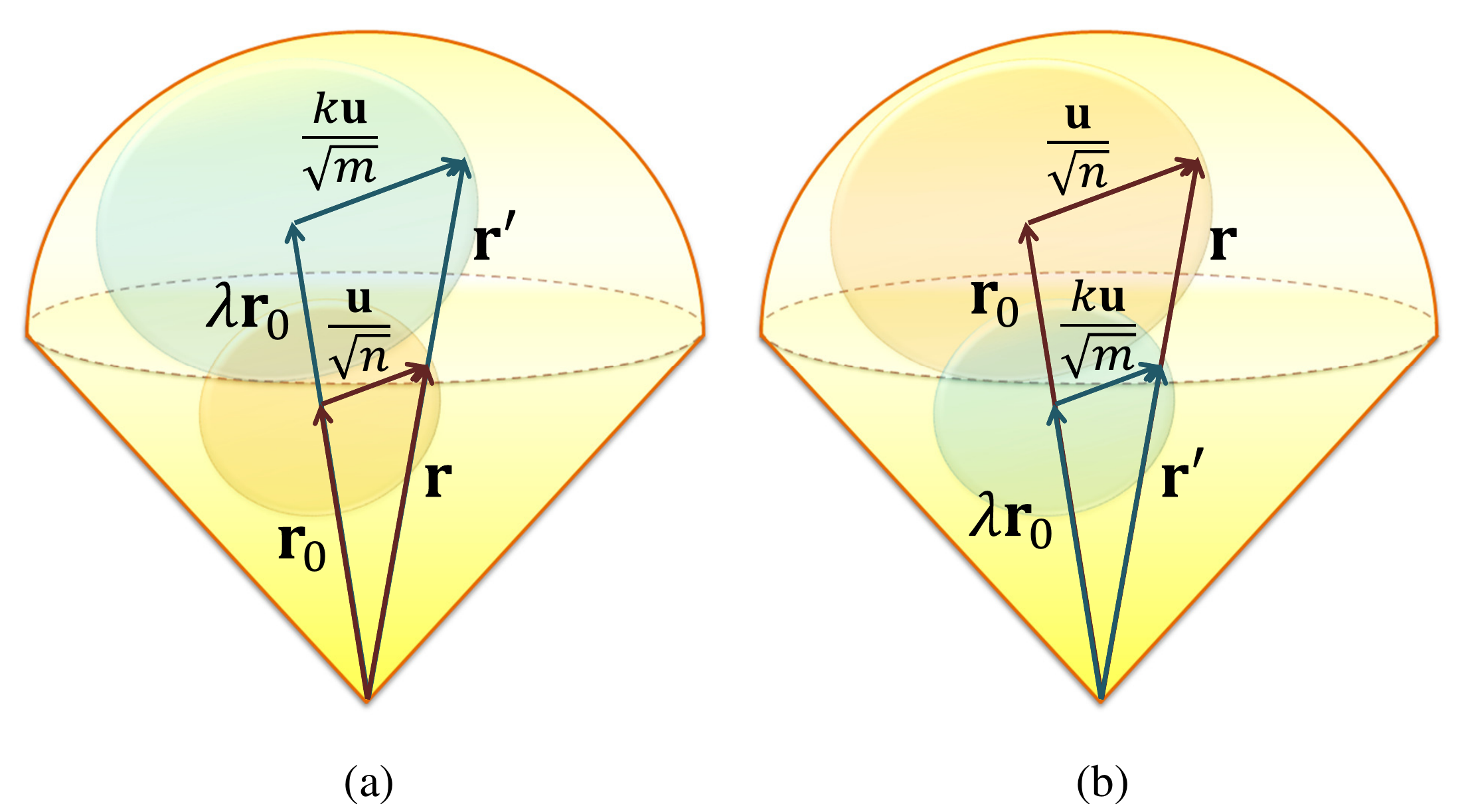}
\caption{(Color online) Schematic Bloch-sphere geometry for (a) purification and (b) dilution of qubits. A change in Bloch vector magnitude by factor $\lambda$ is reflected by a change in the corresponding Gaussian state displacement. The two are then related by $\lambda=k\sqrt{n/m}$.}
\label{figbloch}
\end{figure}

\smallskip

{\it Step 1: Localisation}. We are given $n$ identical qubits in an arbitrary mixed state $\rho_{\textbf{r}}$. We measure a small proportion $n^{1-\epsilon}\ll n$ of the qubits, to obtain a rough estimator $\rho_{\textbf{r}_{0}}$ of the state. By standard concentration results, with asymptotically vanishing probability of error, the actual state is in a local neighbourhood of $\rho_{\textbf{r}_{0}}$ of size $n^{-1/2+\epsilon}$, so that the remaining
qubits can be parametrised as in the local model \eqref{eq.local.qubit.model.}.
In the same time, the target single-system output state $\rho_{\lambda \textbf{r}}$ belongs to the local model
$$
\mathcal{Q}_{m}:=
\left\{
\rho^{\otimes m}_{\lambda\textbf{r}_0 + \frac{\textbf{u}^{\prime}}{\sqrt{m} }} :
\|\textbf{u}^{\prime}\| \leq \lambda \Lambda^{1/2-\epsilon} m^{\epsilon}
\right\}
$$
with local parameter (see Figure \ref{figbloch})
$$
\textbf{u}^{\prime}= \lambda \sqrt{\frac{m}{n}} \textbf{u} =
\lambda \Lambda^{1/2} \textbf{u} := k\textbf{u}.
$$
{\it Step 2: Transfer to the Gaussian state}. We apply the map $T_{n}$ to the
 qubits and obtain a classical random variable and a single-mode CV system whose states are approximately Gaussian [see \eqref{eq.gaussian.model}]
 \begin{equation}\label{eq.gaussian.input}
 N_{\textbf{u}} \otimes \Phi_{\textbf{u}} =
 N(u_{z} ,1- \|\textbf{r}_{0} \|^{2}) \otimes W_{\alpha} \Phi^{s} W_{\alpha}^{\dagger}.
\end{equation}

{\it Step 3: Optimal Gaussian purification (amplification)}.
Since $\mathcal{Q}_{m}$ is a local model around $\lambda\textbf{r}_{0}$, the corresponding parameter of the associated Gaussian model is
$$
s_{2}:= \frac{1- \lambda\|\textbf{r}_{0}\|}{1+\lambda\|\textbf{r}_{0}\|}.
$$
and the family of Gaussian states is
\begin{equation}\label{eq.gaussian.output}
N(ku_{z}, 1- \|k\textbf{r}_{0} \|^{2}) \otimes W_{k\alpha } \Phi^{s_{2}}W_{k\alpha}^{\dagger}.
\end{equation}
In this step we attenuate (amplify) the Gaussian state \eqref{eq.gaussian.input} in order to map it into, or close to \eqref{eq.gaussian.output}, as described in Section~\ref{sec2}.  This means that we apply the optimal channel $K^\star$ defined in section \ref{sec2a} to the classical component $N_{\bf u}$, and the optimal quantum attenuation (amplification) channel $P^{\star}$ defined in Theorem \ref{theorem21}, to the quantum Gaussian component $\Phi_{\bf u}$.

{\it Step 4: Mapping back to the qubits}. We apply the channel $S_{m}$ to the classical variable and the output of the attenuation (amplification) channel to obtain a state of $m$ qubits in the neighbourhood of the state
$\rho_{\lambda\textbf{r}_{0}}$.

\smallskip

\begin{figure}[tb]
\includegraphics[width=6cm]{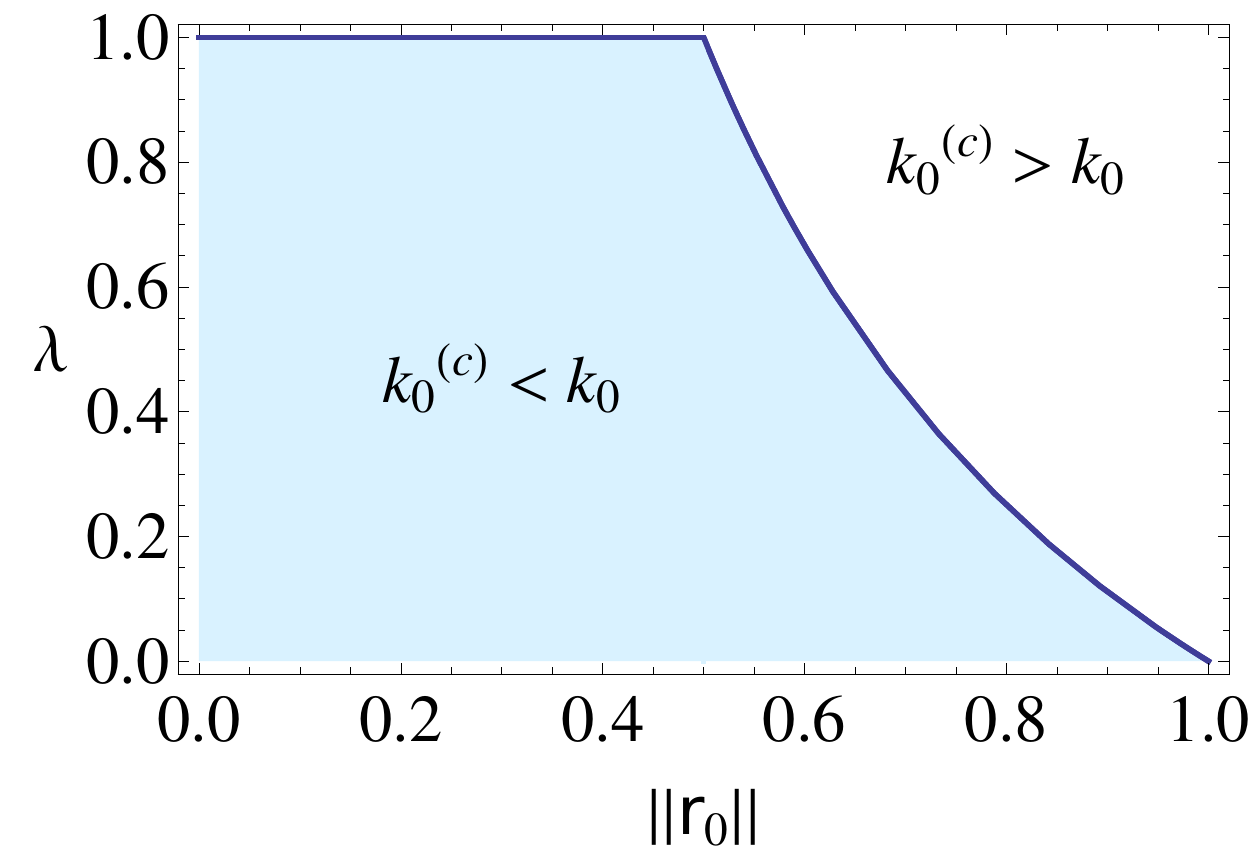}
\caption{(Color online) Parameter range defining the different regimes of qubit dilution. In the shaded region below the boundary curve $\lambda = \tilde\lambda(\|\textbf{r}_0\|)$ [\eq{polybart}], the threshold $k_0^{(c)}$ for a nonzero classical Gaussian risk is smaller than the corresponding threshold $k^{{\rm amp}}_0$ for a nonzero quantum Gaussian risk, and Cases 3 and 4 of Theorem \ref{theorem32} apply for determining the optimal FoM for the qubit problem. Above the boundary curve, the situation is reversed, and Cases 2 and 4 of Theorem \ref{theorem32} apply instead.}
\label{figpoly}
\end{figure}

By composing the channels employed in steps 2--4 we obtain the overall channel
$$
Q_{n}: = S_{n} \circ (K^{\star}\otimes P^{\star})\circ T_n,
$$
which will be shown to be optimal. Recall that for $k\leq k_0$ [see \eq{koeqn}], the quantum component  $\Phi^{s_2}_{k\alpha}$ of the Gaussian target state can be prepared exactly. The same is true for the classical component when $k \leq k_0^{(c)}$, where
\begin{equation}\label{k0c}
k^{(c)}_0 = \sqrt{\frac{1-\lambda^2\|\textbf{r}_0\|^2}{1-\|\textbf{r}_0\|^2}}
\end{equation}
is obtained by substituting $V_1=1-\|\textbf{r}_0\|^2$, $V_2=1-\lambda^2\|\textbf{r}_0\|^2$ for the variances in (\ref{cr}).

This means that the total risk has different expressions over the following three intervals: it is zero when
$0<k\leq {\rm min}(k_0 ,k_0^{(c)}) $, it has one classical or quantum contribution for
${\rm min}(k_0 ,k_0^{(c)})< k \leq {\rm max}(k_0 ,k_0^{(c)} )$ and has both quantum and classical contributions for
$k> {\rm max}(k_0 ,k_0^{(c)} )$. For purification (corresponding to Gaussian attenuation), the ordering $0<k_0^{\rm att}<k^{(c)}<1$ always holds, so the middle interval has a quantum contribution. However, for dilution (corresponding to Gaussian amplification), the ordering of $k_0$ and $k^{(c)}_0$ depends on the parameters $\|\textbf{r}_0\|$ and $ \lambda$. In particular, we see the appearance of a boundary which demarcates the two separate regimes of dilution, each defined by whether classical or quantum contributions to the risk take place first (see Fig.~\ref{figpoly}). Namely, $k_0^{(c)}<k^{{\rm amp}}_0$ for
\begin{equation}\label{polybart}
\lambda < \tilde\lambda(\|\textbf{r}_0\|) \equiv \min\left\{1,\,\frac{1-\|\textbf{r}_0\|}{\|\textbf{r}_0\|}\right\}\,,
\end{equation}
and $k_0^{(c)} \ge k^{{\rm amp}}_0$ otherwise. Notice that  inequality (\ref{polybart}) is always satisfied for $\|\textbf{r}_0\|\leq1/2$ for all values of $\lambda$.

\begin{figure}[t!]
\includegraphics[width=8.5cm]{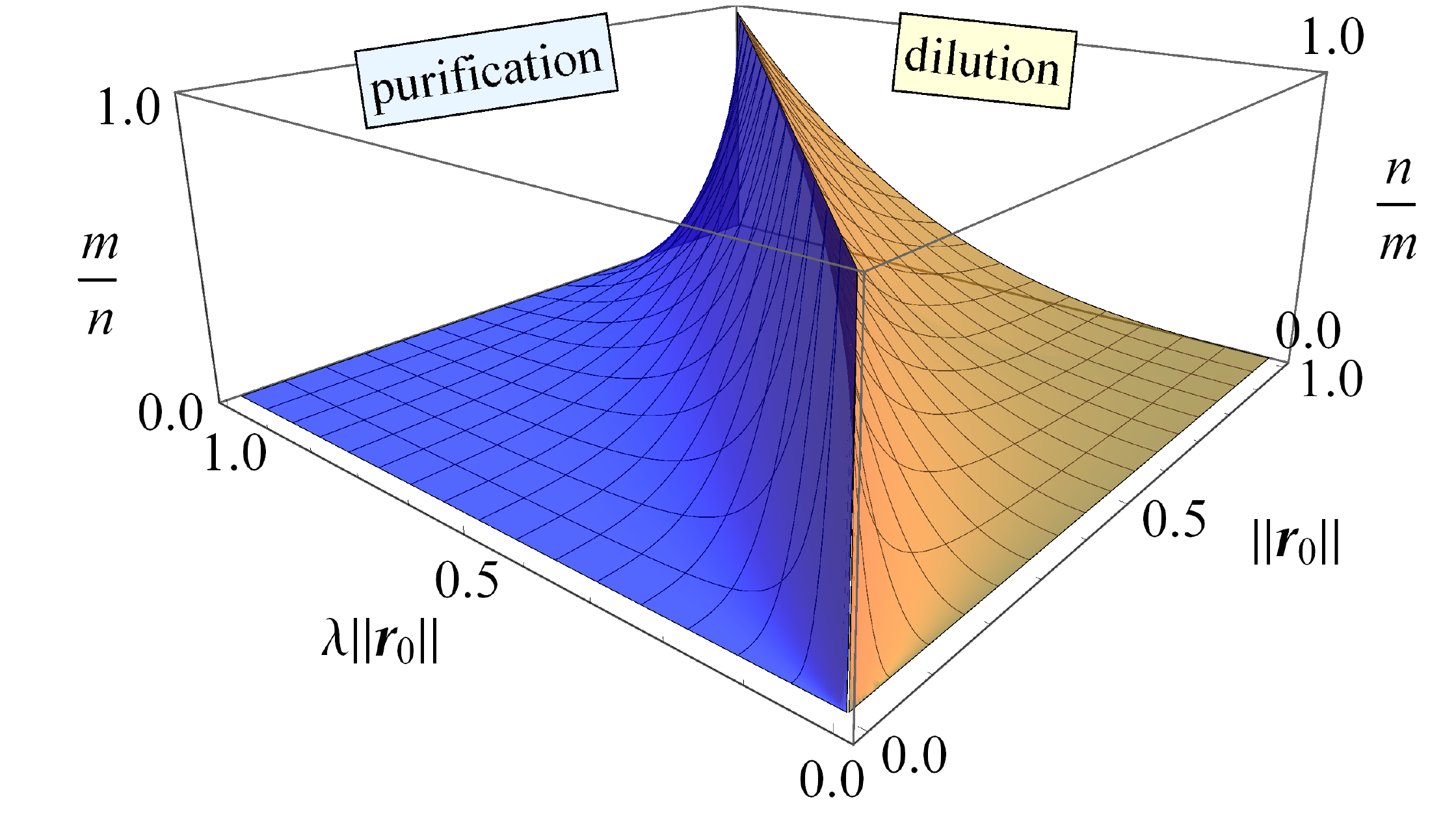}
\caption{(Color online) Optimal input ($n$) vs output ($m$) rates $\Lambda_0$ of qubit production for the processes of perfect purification and dilution, [Eq.~(\ref{eqrate})], plotted versus the Bloch vector lengths before ($\|\textbf{r}_0\|$) and after ($\lambda \|\textbf{r}_0\|$) the protocols. The left (right) side of the three-dimensional surface, corresponding to the region $\lambda > 1$ ($\lambda <1$), represents the optimal rate $m/n$ ($n/m$) for mixed qubit purification (dilution).}
\label{figrates}
\end{figure}

The relation between the output qubit rate $\Lambda=m/n$ and the constants $\lambda,k$ can be inferred from the geometry of
the Bloch sphere (see Fig.~\ref{figbloch})
\begin{equation}
\Lambda(\|\textbf{r}_0\|,\lambda,k)=\frac{k^2}{\lambda^2}.
\label{bsg}
\end{equation}
In particular, the maximum output rates for which the asymptotic risk is zero are obtained by setting: $k=k_0^{\rm att}$ for purification, and $$k=\left\{
      \begin{array}{ll}
        k_0^{(c)}, & \lambda<\tilde\lambda(\|\textbf{r}_0\|)\,, \\
        k_0^{\rm amp}, & \hbox{otherwise,}
      \end{array}
    \right.
$$
for dilution.

\begin{figure*}[tb]
\subfigure[\label{figpurclass}]
{\includegraphics[width=5.5cm]{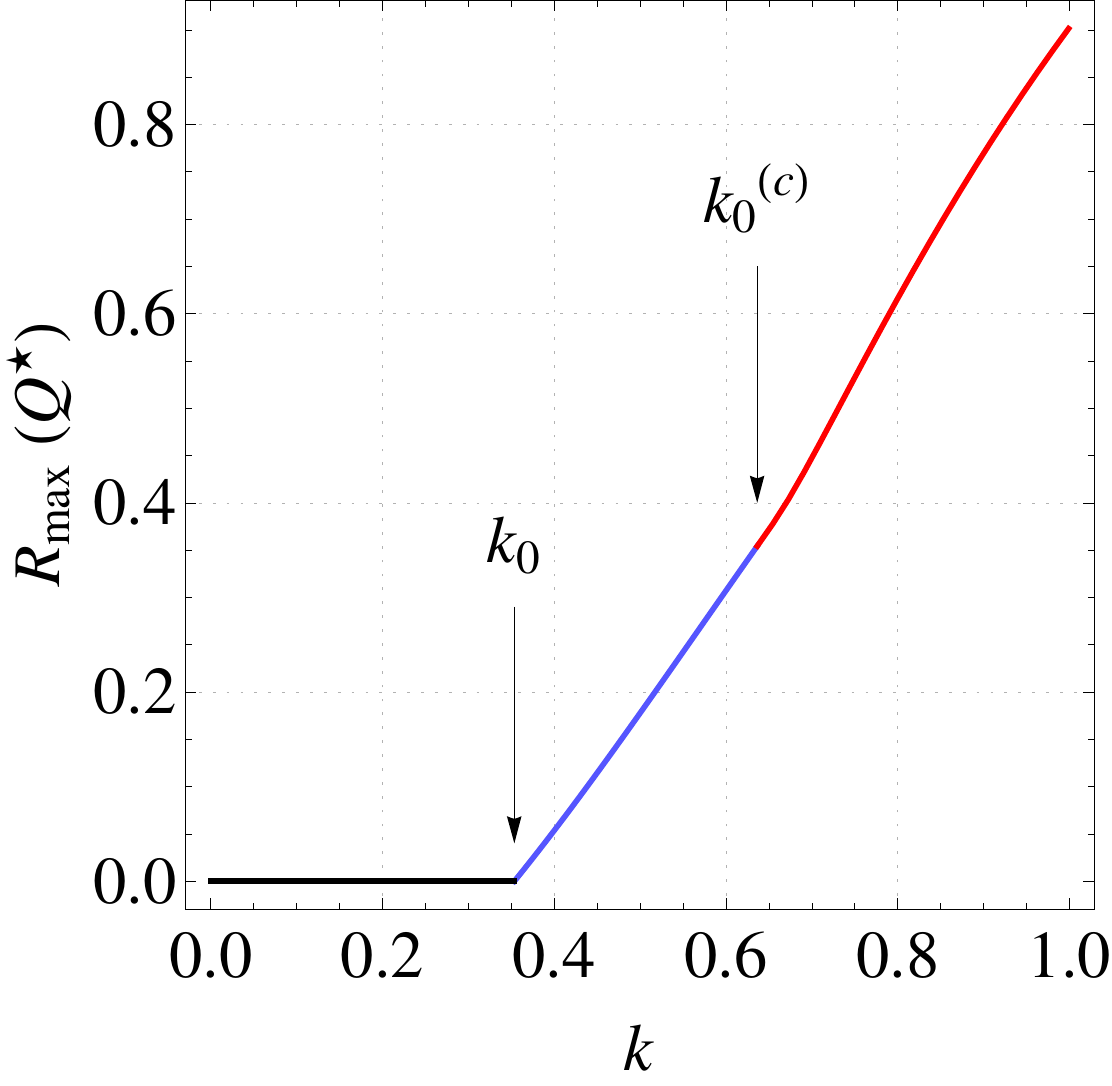}} \hspace{5mm}
\subfigure[\label{figampclass}]
{\includegraphics[width=5.5cm]{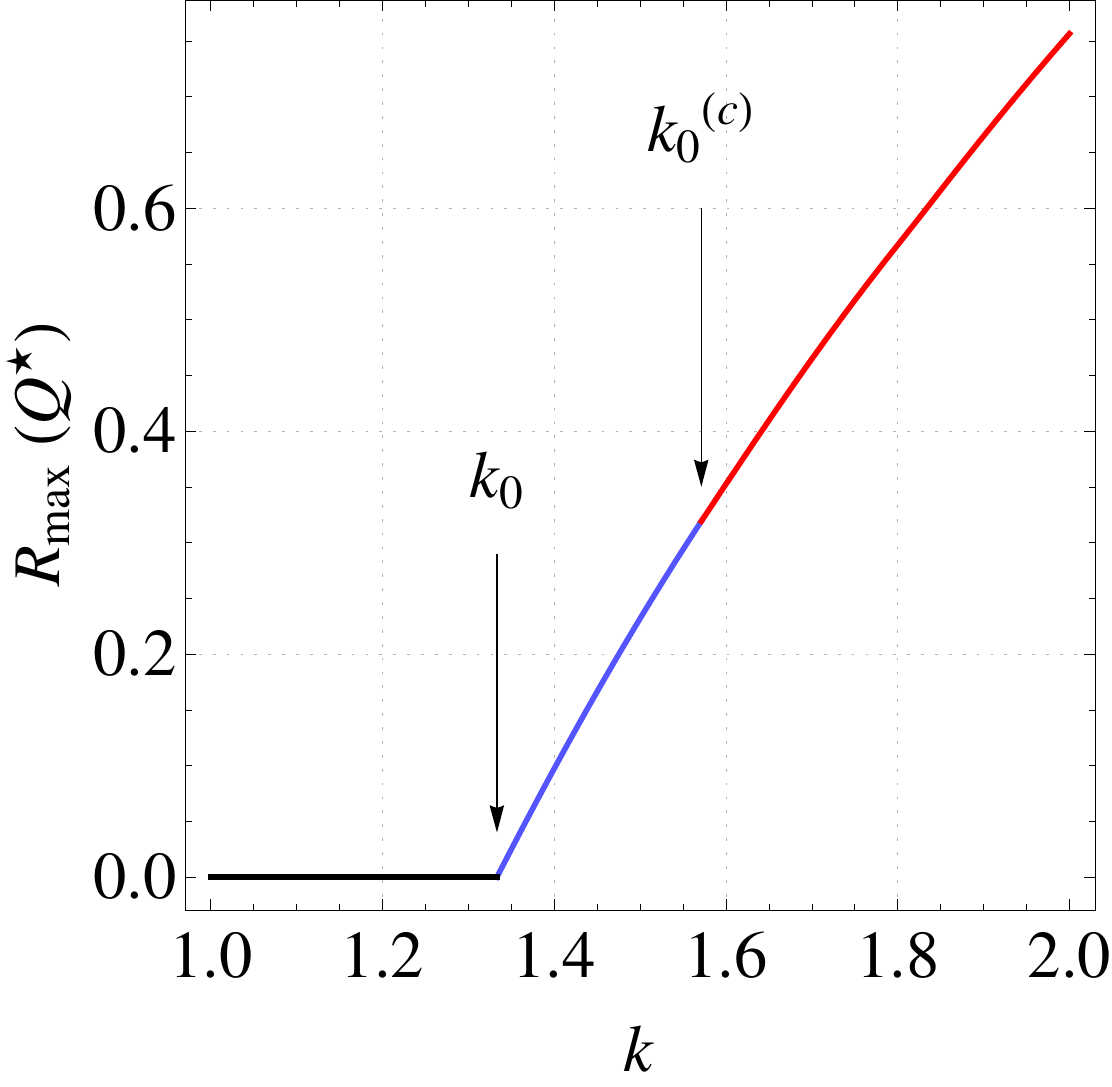}} \hspace{5mm}
\subfigure[\label{figampclass2}]
{\includegraphics[width=5.5cm]{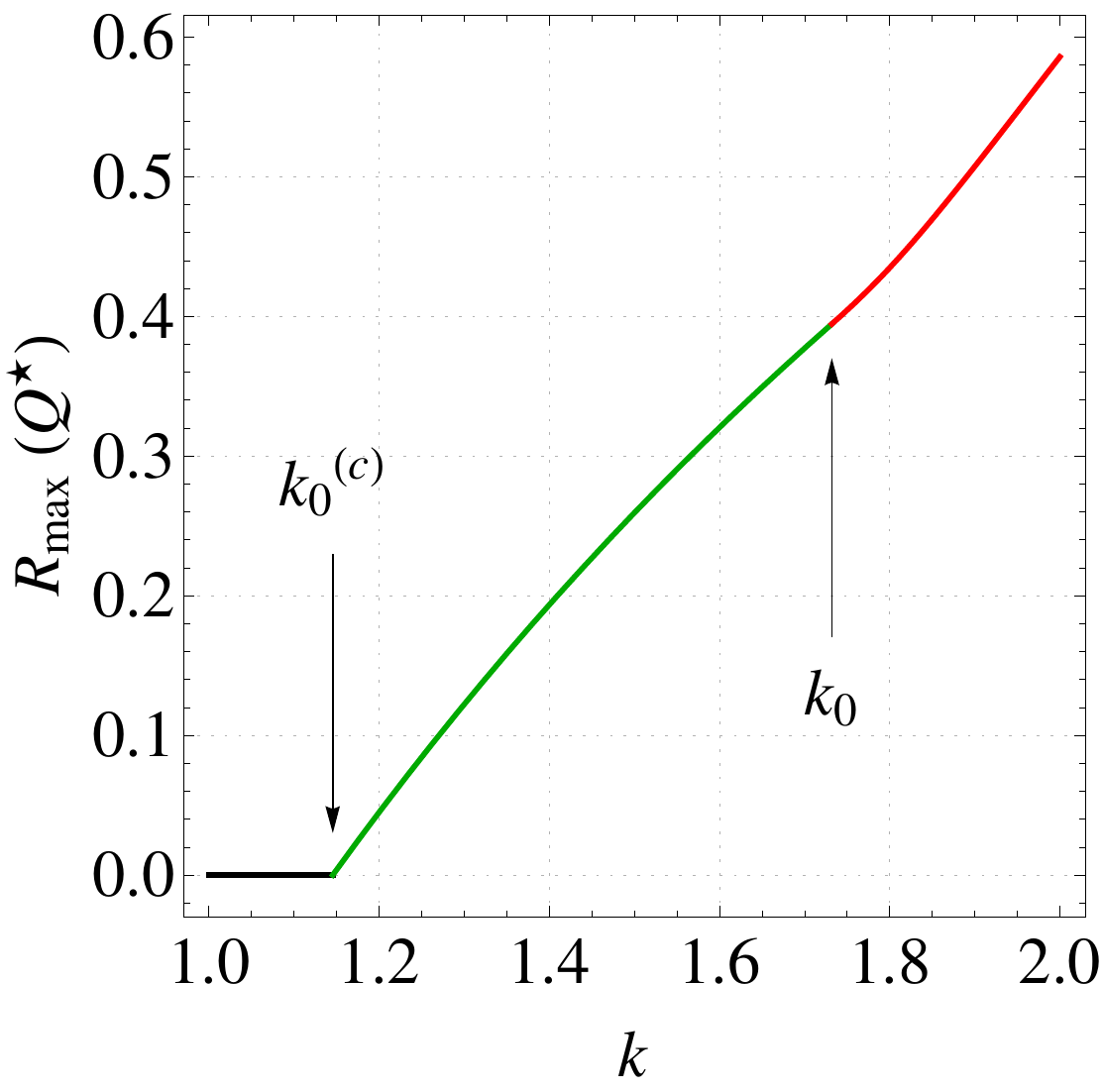}}
\caption{(Color online)
Plot of the minimax risk ${R}_{\max}(Q^\star)$ [Eqs.~(\ref{riskqubeq}--\ref{riskqubeqcl})] for (a) optimal purification with $\|\textbf{r}_0\|=1/3,  \lambda\|\textbf{r}_0\|=4/5$,  (b) optimal dilution with $\|\textbf{r}_0\|=4/5,  \lambda\|\textbf{r}_0\|=1/3$, and (c) optimal dilution with $\|\textbf{r}_0\|=1/2,  \lambda\|\textbf{r}_0\|=1/8$,  of $n$ qubits, as a function of the local scale factor $k$. In panels (a) and (b), the risk is zero for $k<k_0$ (black line), it is given by the quantum Gaussian FoM for $k_0\leq k <k_0^{(c)}$ (blue line, corresponding to Case 2 of Theorem~\ref{theorem32}), and it is then given by the joint quantum and classical contributions for $k \geq k_0^{(c)}$ (red line, corresponding to Case 4 of Theorem~\ref{theorem32}). In panel (c), which describes a dilution characterised by a choice of parameters within the shaded region of Fig.~\ref{figpoly}, we have instead that the risk is zero for $k<k_0^{(c)}$ (black line), it is given by the classical Gaussian FoM for $k_0^{(c)}\leq k <k_0$ (green line, corresponding to Case 3 of Theorem~\ref{theorem32}), and it is then given by the joint quantum and classical contributions for $k \geq k_0$ (red line, corresponding to Case 4 of Theorem~\ref{theorem32}).}
  \label{figclass}
\end{figure*}

Explicitly,
\begin{eqnarray}
\label{eqrate}
\Lambda_0^{\rm pur}(\|\textbf{r}_0\|,\lambda) &=&  \frac{ \lambda^{-1} - \|\textbf{r}_0\|}{\lambda^2 (1- \|\textbf{r}_0\|)}\,; \nonumber \\ \\
\nonumber
\Lambda_0^{\rm dil}(\|\textbf{r}_0\|,\lambda) &=& \left\{
      \begin{array}{ll}
        \frac{\lambda^{-2}- \|\textbf{r}_0\|^2}{1- \|\textbf{r}_0\|^2}, & \ \ \lambda<\tilde\lambda(\|\textbf{r}_0\|)\,, \\ & \\
        \frac{\|\textbf{r}_0\| + 1/\lambda}{\lambda^2 (\|\textbf{r}_0\| + 1)}, & \ \ \hbox{otherwise.}
      \end{array}
    \right.
\end{eqnarray}
The optimal rates \eqref{eqrate} are plotted in Fig.~\ref{figrates}.

 We can now state main result of this section whose proof is given in Appendix A.

\begin{theorem}\label{theorem32}
The sequence of purification (dilution) maps
\begin{equation}
Q^{\star}_n := S_m\circ (K^{\star}\otimes P^{\star}) \circ T_n
\end{equation}
is locally asymptotically minimax.

We distinguish four cases for the minimax risk.

\textbf{Case} $\mathbf{1}$: Zero risk. If $k< {\rm min}(k_0, k_0^{(c)})$ then the minimax risk is zero. The optimal output rate $\Lambda_0$ is given in \eqref{eqrate}.

\textbf{Case} $\mathbf{2}$: Quantum contribution.
If $k_0\leq k \leq k^{(c)}_0$, then the purification (dilution) minimax risk at
$\textbf{r}_0$ is equal to the risk of the optimal Gaussian attenuation  (amplification) scheme (\ref{Rmaxpur})
\begin{equation}\label{riskqubeq}
R_{\rm minmax}(\textbf{r}_0,\lambda,\Lambda)=R_{\rm minmax}(k,s_1,s_2).
\end{equation}

\textbf{Case} $\mathbf{3}$: Classical contribution. If $k^{(c)}_0\leq k\leq k_0$, then the dilution minimax risk at $\textbf{r}_0$ is equal to the risk of the optimal classical Gaussian amplification scheme (\ref{clrisk}):
\begin{equation}
R_{\rm minmax}(\textbf{r}_0,\lambda,\Lambda)=R_{\rm minmax}(V_{1},V_{2},k).
\end{equation}

\textbf{Case} $\mathbf{4}$: Classical and quantum contributions. If $k> {\rm max}(k_0,k^{(c)}_0)$, then the purification (dilution) minimax risk is
\begin{align}\label{riskqubeqcl}
&R_{\rm minmax}(\textbf{r}_0,\lambda,\Lambda)\nonumber \\
&=\int{dx}\sum_n\Bigg|\frac{\exp\bigg(-\frac{x^2}{2k^2(1-\|\textbf{r}_0\|^2)}\bigg)(1-\tilde{s})\tilde{s}^n}{\sqrt{2\pi k^2 \big(1-\|\textbf{r}_0\|^2\big)}}\nonumber \\
&-\frac{\exp\bigg(-\frac{x^2}{2(1-\lambda^2\|\textbf{r}_0\|^2)}\bigg)(1-s_2)s^n_2}{\sqrt{2\pi \big(1-\lambda^2\|\textbf{r}_0\|^2\big)}}\Bigg|
\end{align}
where $\tilde{s}$ takes the values given in (\ref{tildes}) in the case of attenuation (for qubit purification) and  amplification (for qubit dilution) respectively.

\end{theorem}

The optimal minimax risk for purification and dilution of qubits is plotted in Fig.~\ref{figclass} as a function of $k$.

\section{Conclusions}\label{sec4}

We have solved the practically relevant problem of optimal attenuation and amplification of displaced Gaussian states, with respect to the maximum norm-one distance FoM. As expected, the optimal channels are implemented by the beamsplitter and parametric amplifier respectively, where the ancillary state is provided by the vacuum in both cases. This solution was then used in conjunction with LAN, to construct optimal purification and dilution channels for ensembles
of i.i.d. qubits  as formulated in Theorem \ref{theorem32}.

In the Gaussian case, we give an explicit expression of the FoM as a function
of the variance parameters $s_{1}$ and $s_{2}$ of input and output states and the attenuation (amplification) factor $k$. In particular we identify the optimal value $k_{0}(s_{1},s_{2})$ for which the protocol achieves the target state exactly. Similarly, in the multiple qubits case, we derive the FoM as a function of the input and output purity and the asymptotic input/output rate $\Lambda$,
and identify the optimal rate $\Lambda_{0}$ for which the the protocol achieves the target collective state exactly. Both classical and quantum Gaussian contributions to the risk have to be taken into account to calculate the maximum rates, and to provide the optimal FoM for purification and dilution of multiple qubits, in the parameter range where the procedures cannot be accomplished perfectly.

An interesting future project is to extend the techniques used in this paper to tackle the general problem of asymptotically optimal channel inversion for arbitrary channels on finite dimensional systems. Such a study may provide efficient strategies to counteract the effect of various types of noise and decoherence processes, beyond the depolarising channel considered in the present work.

\acknowledgments{
M\u{a}d\u{a}lin Gu\c{t}\u{a} was supported by the EPSRC Fellowship EP/E052290/1.}

\appendix
\section{Proofs}

\subsection*{Proof of Theorem \ref{theorem21}}
\begin{proof}
As the proof follows the lines of similar results in \cite{GM,GBA} we will briefly sketch the main ideas. A covariance argument \cite{C,O} shows that one may restrict the attention to channels which are displacement-covariant, in the sense that $P(W_\xi \Phi W^\dagger_\xi) = W_\xi P(\Phi)W^\dagger_\xi$ for any input state $\rho$. For such channels the risk is independent of $\alpha$ and
$$
{R}_{\max}(k,s_1,s_2,P) = \|P(\Phi^{s_1}) - \Phi^{s_2}\|_1
$$
In the case of attenuation $(k<1)$ such channels are described by the linear transformation
$$
c_{\rm att} = ka + \sqrt{1-k^2}b
$$
where $a$ is the input mode, $c_{\rm att}$ is the output and  $b$ is an ancillary mode prepared in a state $\tau$. Since the channel is completely characterised by the state $\tau$, we will denote it by $P_\tau$. Similarly, for amplification
$(k>1)$ the output of the channel $P_\tau$ is the mode
$$
c_{\rm amp} = ka + \sqrt{k^2-1}b^\dagger,
$$ with $b$ prepared in the state $\tau$. By a second covariance argument with respect to phase rotations, and taking into account that $\Phi$ is invariant under phase  rotations, we obtain that $\tau$ can be taken to be phase-invariant, i.e. it is a mixture of Fock states $\tau=\sum_i\tau_i|i\rangle\langle i|$. In this case the output state will be diagonal in the Fock basis and we write $P_\tau(\Phi^{s_1}) = \sum_ip^\tau_i|i\rangle\langle i|$, and in particular $p^\omega$ corresponds to the output state when the ancilla is the vacuum. Similarly, we denote the coefficients of the Gaussian state $\Phi^{s_1}$ and $\Phi^{s_2}$ by $p_i=(1-s_1)s^i_1$ and $q_i=(1-s_2)s^i_2$. The proof reduces now to showing that, for any $\tau$,
\begin{equation}\label{eq.norm.ineq}
\|p^\tau - q\|_1 \geq \|p^\omega - q\|_1.
\end{equation}
The key to proving this statement is the concept of stochastic ordering, whose definition we recall:
\begin{definition}[Stochastic Ordering]
Let $p=\{p_l:l\in\mathbb{N}\}$ and $q = \{q_l: l\in\mathbb{N}\}$ be two probability distributions over $\mathbb{N}$. We say that $p$ is
stochastically smaller than $q$ $(p\preceq q)$ if
\begin{equation}
\sum^m_{l=0} p_l \geq \sum^m_{l=0} q_l, \quad \forall m \geq 0
\end{equation}
\end{definition}
The following lemma holds for both  the purification and the amplification scenarios:
\begin{lemma}\label{lemma42}
For any state $\tau$ the following stochastic ordering holds
\begin{equation}
p^{\omega} \preceq p^{\tau}.
\label{vacstoch}
\end{equation}
\end{lemma}
\begin{proof}
We treat the attenuation and amplification separately but the idea is the same in both cases: we reduce the statement about stochastic ordering to a simpler one where the input mode is in the vacuum.

\noindent
{\it Attenuation.} We write the input mode as $a = \cosh(t)a_1 + \sinh(t)a^\dagger_2$ with $a_{1,2}$ two fictitious modes in the vacuum state, and $\tanh^2(t) := s_{1}$, which ensures that the state of $a$ is $\Phi^{s_{1}}$.
Let $\tilde{t}$ be such that $\sinh(\tilde{t})=k\sinh(t)$ and denote
$$
T = \sqrt{1 - \frac{(1 - k^2)}{\cosh^2(\tilde{t})}}, \qquad
R = \frac{\sqrt{1 - k^2}}{\cosh(\tilde{t})}.
$$
Then $c_{pur}=ka+\sqrt{1-k^{2}}b $ can be written as
\begin{eqnarray*}
c_{pur}& =& \cosh(\tilde{t})(Ta_1 + Rb) + \sinh(\tilde{t})a^\dagger_2 \\
&=& \cosh(\tilde{t})\tilde{b} + \sinh(\tilde{t})\tilde{a}^\dagger
\end{eqnarray*}
where  $\tilde{b} := Ta_1 + Rb$ and $a_{2} $ was relabelled $\tilde{a}$.
The state of the mode $\tilde{b}$ is given by
\begin{eqnarray*}
\tilde{\tau} &=& \sum^\infty_{k=0} \tau_k \sum^k_{p=0} \binom{k}{p} T^{2(p-k)} R^{2k} |p\rangle\langle p|\\
 & =& \sum^\infty_{p=0} \tilde{\tau}_p |p\rangle\langle p|
\end{eqnarray*}
so  $\tilde{b}$ is in the vacuum state if and only if $b$ is in the vacuum.
Thus it suffices to prove the stochastic ordering statement for the mode
$c_{\rm att}$ written as a combination of $\tilde{b}$ and $\tilde{a}$ for an arbitrary diagonal state $\tilde{\tau}$ of $\tilde{b}$ and $\tilde{a}$ in the vacuum. Furthermore, since stochastic ordering is preserved under convex combinations, it suffices to prove the statement for any {\it pure} diagonal state $\tilde{\tau} = |k\rangle\langle k|$, $k\neq 0$. In this case the state of $c_{\rm att}$ is given by
\begin{eqnarray*}
\rho^{\rm out}_{\rm att} &=& e^{-2g(k+1)}\sum^\infty_{l=0} \Gamma^{2l} \binom{l+k}{k} |l+k\rangle\langle l+k| \\
&:=& \sum^m_{l=0} d^{(k)}_l |l+k\rangle\langle l+k|
\end{eqnarray*}
where $\Gamma = \tanh({\tilde{t}})$ and $e^g = \cosh(\tilde{t})$.  The stochastic ordering now reduces to showing that $\sum^m_{l=0} d^{(0)}_l \geq \sum^m_{l=0} d^{(k)}_l$ for all $m$. With the notation $\gamma = \Gamma^2$, we get
\begin{eqnarray*}
\sum^{p+k}_{l=0}d^{(k)}_l &=& (1-\gamma)^{k+1} \sum^p_{l=0}\gamma^l\binom{l+k}{k}
\\ &\leq&  1 - \gamma^{p+1} \sum^k_{r=0} (1-\gamma)^r\gamma^{k-r}\binom{k}{r}
\\ &=& 1-\gamma^{p+1} = \sum^p_{l=0} d^{(0)}_l.
\end{eqnarray*}

{\it Amplification.} As before we write $a = \cosh(t)a_1 + \sinh(t)a^\dagger_2$ and define $\tilde{t}$ by
$\cosh({\tilde{t}}) = k\cosh(t)$ and the beamsplitter coeficients
$$
T = \sqrt{1 - \frac{(1 - k^2)}{\sinh^2(\tilde{t})}} \qquad
R = \frac{\sqrt{1 - k^2}}{\sinh(\tilde{t})}.
$$
The output mode is now
\begin{eqnarray*}
c_{\rm amp} &=& \sinh(\tilde{t})(Rb^\dagger + Ta^\dagger_2) + \cosh(\tilde{t})a_1 \\
&=& \sinh(\tilde{t})\tilde{b}^\dagger + \cosh(\tilde{t})\tilde{a}
\end{eqnarray*}
where we have relabelled $a_1$ by $\tilde{a}$ and introduced the
mode $\tilde{b}=Rb^\dagger + Ta^\dagger_2$. As before, the state of $\tilde{b}$ is the vacuum if and only if $b$ is in the vacuum state, so it suffices to verify the statement for the state $\tilde{\tau}=|k\rangle \langle k|$ in which case the output state is
$$
\rho^{\rm out}_{\rm amp} = e^{-2g(k+1)}\sum^\infty_{l=0} \Gamma^{2l} \binom{l+k}{k} |l\rangle\langle l| = \sum^m_{l=0} d^{(k)}_l|l\rangle\langle l|
$$
The relation $p^\omega \preceq p^\tau$ now follows from
\begin{eqnarray*}
\sum^{p}_{l=0}d^{(k)}_l &=& (1-\gamma)^{k+1} \sum^p_{l=0}\gamma^l\binom{l+k}{k} \\
&\leq& 1-\gamma^{p+1}= \sum^p_{l=0} d^{(0)}_l.
\end{eqnarray*}
This ends the proof of Lemma \ref{lemma42} for both cases.
\end{proof}

The following lemma completes  the proof of Theorem \ref{theorem21} by transforming the stochastic ordering into the desired norm inequality \eqref{eq.norm.ineq}. Its proof \cite{GM,GBA} uses the fact that
$q \preceq  p^{\omega} $ which is equivalent to the fact that $P^{\star}(\Phi^{s_{1}})$ is more noisy that $\Phi^{s_{2}}$. The latter is satisfied for $k\geq k_{0}$ as assumed in the theorem.
\begin{lemma}
Let $ p^{\prime}$ be a discrete probability distribution such that
$p^{\omega}\preceq p^{\prime}$. Then
\begin{equation}
\|p^\prime - q\|_1 \geq \|p^\omega - q\|_1.
\end{equation}
\end{lemma}

%
\end{proof}

\subsection*{Proof of Lemma \ref{lemma22}}
We use the notations introduced in the proof of Theorem~\ref{theorem21}. By expressing the quadrature variance of the input mode $a$ in terms of $t$ and $s_1$ we obtain $\sinh^2t = \frac{1}{e^{s_1} - 1}$. According to Theorem \ref{theorem21} the output state of the optimal channel
$P^{\star}$ is the Gaussian state
$$
P^{\star}( \Phi^{s_{1}})= \Phi^{\tilde{s}} =
e^{-2g}\sum^\infty_{l=0}(1-e^{-2g})^l|l\rangle\langle l|
$$
with $g$ taking different values in the attenuation and amplification cases. For the geometric distributions $p^{\omega}$ and $q$ we have
$$
\| \Phi^{\tilde{s}}- \Phi^{s_{2}} \|_{1}= \|p^{\omega}-q\|_{1} =
2(\tilde{s}^{m_{0}+1} - s_{2}^{m_{0}+1})
$$
where $m_{0}$ is the largest integer such that $p^{\omega}_{m_{0}} \leq q_{m_{0}}$, more precisely
$$
m_{0}= \lfloor \ln [ (1-\tilde{s})/(1-s_{2})] /\ln(s_{2}/\tilde{s}) \rfloor
$$
It remains to compute the concrete expressions of $\tilde{s}$ and implicitly of  $m_{0}$ for the attenuation and amplification cases. For attenuation, making use of $\sinh\tilde{t} = k\sinh t$, we find
$$
\tilde{s}_{pur} = \frac{s_1 k^2}{1-s_1 + s_1 k^2}.
$$
For amplification, we use $\cosh\tilde{t} = k\cosh t$ and find
$$
\tilde{s}_{amp}= 1-\frac{1-s_1}{k^2}.
$$
%
%

\subsection*{Proof of Theorem \ref{theorem32}}
\begin{proof}
We want to show that $Q^\star:=S_n\circ P^\star \circ T_n$ is the optimal purification or dilution procedure for $n$ i.i.d. qubits. The idea is that, by using LAN, we can show the qubit and Gaussian statistical problems to be equivalent, the Gaussian one (respectively for attenuation and amplification) being solved in Section~\ref{sec2b}, which then allows us to recast the qubit problem in the Gaussian setup with a vanishing difference in the risks. We will consider the four separate cases: zero risk, solely quantum contribution, solely classical contribution, and both classical and quantum contributions. We will then use the Gaussian solution to show that $R_{\max}(\textbf{r}_0,Q^\star,\lambda)$ is less than or equal to the corresponding optimal Gaussian risk, then show that a strict inequality violates the optimality of this optimal solution. We begin by restricting $\textbf{r}$ to the local neighbourhood $\|\textbf{r}-\textbf{r}_0\|\leq n^{-\frac{1}{2}+\epsilon}$. This probability that the state fails to be in this region is $o(1)$ and has no influence on the asymptotic risk (see Lemma 2.1. in \cite{GJK}). We are now able to apply LAN, which maps input states $\rho^{\otimes n}_{\textbf{r}}$ close to some Gaussian state, say $\tilde{\Phi}_{\textbf{u}}$, via the channel $T_n$. We now consider the individual cases, which are each slight variations on the same proof:

\textbf{Case}$\mathbf{1}$: $k<{\rm min}(k_0,k_0^{(c)})$. In this case both classical and quantum Gaussian channels have zero risk, so the asymptotic qubit risk is zero.

\textbf{Case} $\mathbf{2}$: $k_0\leq k\leq k^{(c)}_0$. In this instance, the risk receives only a quantum contribution.
Using contractivity of the CP maps $S_m$ and $P^\star$, the LAN convergence, and the fact that in this regime
$K^\star N_{\bf u} = N_{{\bf u}^\prime}$ we obtain
\begin{align}
&R(Q^\star_n, \textbf{r},\lambda)\nonumber\\
&\qquad =\|\rho^m_{\textbf{u}'}-Q^\star_n(\rho^n_{\textbf{u}})\|_1\nonumber \\
&\qquad\leq\|\rho^m_{\textbf{u}^\prime}-S_m(\Phi^{s_2}_{\textbf{u}^\prime} \otimes N_{{\bf u}^\prime})\|_1\nonumber \\
&\qquad ~+\|S_m(\Phi^{s_2}_{\textbf{u}'} \otimes N_{{\bf u}^\prime}) )-S_m\big(P^\star \otimes K^\star \big(T_n(\rho^n_{\textbf{u}})\big)\big)\|_1\nonumber \\
&\qquad\leq\|\Phi^{s_2}_{\textbf{u}^\prime} \otimes N_{{\bf u}^\prime} -P^\star \otimes K^\star \big(T_n(\rho^n_{\textbf{u}})\big)\|_1+o(1)\nonumber \\
&\qquad = \|\Phi^{s_2}_{\textbf{u}^\prime} \otimes N_{{\bf u}^\prime} -P^\star \otimes K^\star \big( \Phi^{s_1}_{\textbf{u}} \otimes N_{\bf u}\big)\|_1 +o(1) \nonumber\\
&\qquad = \|\Phi^{s_2}_{\textbf{u}^\prime}-P^\star\big( \Phi^{s_1}_{\textbf{u}})\|_1 +o(1)\nonumber \\
&\qquad = R_{\rm minmax}(s_1,s_2,k) + o(1)\label{eq.upperbound}
\end{align}
where $R_{\rm minmax}(s_1,s_2,k) $ is the minimax risk for the quantum Gaussian problem, obtained in Theorem \ref{theorem21}.
By taking supremum over $\|{\bf u}\|<n^\epsilon$ we get
$$
R_{\rm max}(Q^\star_n, \textbf{r}_0,\lambda) \leq R_{\rm minmax}(s_1,s_2,k).
$$
which implies that
$$
R_{\rm minmax}(\textbf{r}_0,\lambda) \leq R_{\rm minmax}(s_1,s_2,k).
$$

Next, we show by contradiction that this inequality cannot be strict. Suppose that there exists a sequence of purification or dilution procedures $\tilde{Q}_n$, which act on qubits and satisfies $R_{max}(\tilde{Q}_n\textbf{r}_0,\lambda)\leq R_{\rm minmax}(s_1,s_2,k)- \eta$ for some $\eta>0$ and $n>n_0$. We will use LAN to show that there exists a Gaussian dilution (amplification) channel whose risk is strictly smaller than the minimax risk, which is a contradiction.

The general setup can be seen in (\ref{ocd})
\begin{equation}\begin{CD}
  \rho^{\otimes n}_{\textbf{u}} @>\tilde{Q}_n>> \rho^{\otimes m}_{\lambda\textbf{r}} \\
  @AA{S_n}A @V{T_m}VV       \\
    \Phi^{s_1}_{\textbf{u}} @>\tilde{P}>> \Phi^{s_2}_{k\alpha} \\
   \end{CD}\label{ocd}
   \end{equation}
Here LAN is restricted to a {\it two dimensional} family of rotated qubit states, for which the limit model is quantum Gaussian, with no classical component. Assuming $\Phi^{s_1}_{\textbf{u}}$ is in the domain of applicability of LAN  (which can be effected by an adaptive measurement \cite{GBA}), we get the inequalities
\begin{align}
&\|T_m\circ\tilde{Q}_n\circ S_n(\Phi^{s_1}_{\textbf{u}})-\Phi^{s_2}_{\textbf{u}'}\|\nonumber \\
&\qquad\leq\|\tilde{Q}_n(\rho^n_{\textbf{u}})-\rho^m_{\textbf{u}'}\|_1+\|T_m(\rho^m_{\textbf{u}'})-\Phi^{s_2}_{\textbf{u}'}\|_1\nonumber \\
&\qquad\leq R_{\rm minmax}(s_1,s_2,k) -\eta + o(1).
\end{align}
By taking the limit $n\to \infty$ we get the desired contradiction.

\textbf{Case} $\mathbf{3}$: $k^{(c)}_0\leq k\leq k_0$. This case applies only to dilution and the risk receives only a classical contribution. The proof follows the same steps as the previous case, with the quantum Gaussian replaced by the classical one. The inequality \eqref{eq.upperbound} becomes
\begin{align}
&R(Q^\star_n, \textbf{r} ,\lambda) \nonumber\\
& \qquad =\|\rho^m_{\textbf{u}'}-Q^\star_n(\rho^n_{\textbf{u}})\|_1\nonumber \\
&\qquad\leq \|N(0,1-\lambda^2\|\textbf{r}_0\|^2)-K^\star\big( N(0, 1-\|\textbf{r}_0\|^2)\big)\|_1+o(1)\nonumber \\
&\qquad =  R_{\rm minmax}(V_{1},V_{2},k) + o(1).
\end{align}
where $R_{\rm minmax}(V_{1},V_{2},k)$ is the optimal risk of Eq.~(\ref{clrisk}) and we have identified $V_1=1-\|\textbf{r}_0\|^2$, $V_2=1-\lambda^2\|\textbf{r}_0\|^2$. This implies
$$
R_{\rm minmax}(\textbf{r}_0,\lambda) \leq R_{\rm minmax}(V_{1},V_{2},k).
$$
The equality is obtained by showing that strict inequality would lead to a classical amplification
procedure whose risk is smaller than the minimax risk.

\textbf{Case} $\mathbf{4}$: $k>{\rm max}(k^{(c)}_0, k_0)$. This case applies to both dilution and amplification, and both the quantum and classical channels contribute to the risk.
\begin{align}
&R_{\max}(Q_n^\star,\textbf{r},\lambda)\nonumber\\
&\qquad
= \|\rho^m_{\textbf{u}'}-Q_n^\star(\rho^n_{\textbf{u}})\|_1\nonumber \\
&\qquad \leq\|\Phi^{s_2}_{\textbf{u}'}\otimes N_{{\bf u}^\prime}-
K^\star\otimes P^\star\big( \Phi^{s_1}_{\textbf{u}}\otimes N_{{\bf u}}   \big)\|_1+o(1)\nonumber \\
&\leq R_{\rm minmax}(V_{1},V_{2},s_1,s_2,k)+o(1).\label{a9}
\end{align}
Here $R_{\rm minmax}(V_{1},V_{2},s_1,s_2,k)$ is the minimax norm-one risk for the problem of transforming the Gaussian state $\Phi^{s_1}_{\textbf{u}}\otimes N_{{\bf u}}$ into
$\Phi^{s_2}_{\textbf{u}'}\otimes N_{{\bf u}^\prime}$. Since the quantum and classical components are independent and have different local parameters, the optimal channel is the product $K^\star\otimes P^\star$. The explicit expression of the minimax risk is
\begin{align}
&R_{\rm minmax}(V_{1},V_{2},s_1,s_2,k)\nonumber \\
&\quad = \int{dx}\sum_n\bigg|\frac{\exp\bigg(-\frac{x^2}{2k^2(1-\|\textbf{r}_0\|^2)}\bigg)}{\sqrt{2\pi k^2 \big(1-\|\textbf{r}_0\|^2\big)}}(1-\tilde{s})\tilde{s}^n\nonumber \\
&\qquad-\frac{\exp\bigg(-\frac{x^2}{2(1-\lambda^2\|\textbf{r}_0\|^2)}\bigg)}{\sqrt{2\pi \big(1-\lambda^2\|\textbf{r}_0\|^2\big)}}(1-s_2)s^n_2\bigg|
\label{qacopt}
\end{align}
Finally, the equality in (\ref{a9}) can be proven by contradiction as in case 2.
\end{proof}


\end{document}